\documentclass[letterpaper,11pt]{article}

\usepackage{pslatex} 
\usepackage{amsmath}
\usepackage{amsfonts, mathtools}
\usepackage{epsfig}
\usepackage{graphics}
\usepackage[hang]{subfigure}
\usepackage{color}
\usepackage{fancybox}
\usepackage{cite}
\usepackage{url}
\usepackage{xspace}
\usepackage{fullpage}

\newcommand{\comment}[1]{\marginpar{\sf \tiny \raggedright\sloppy {\bf TODO:} #1}}
\newcommand{\punt}[1]{\relax}
\renewcommand{\comment}[1]{\relax}

\newcommand{\poly}{\mbox{\rm {poly}}\xspace}

\newcommand{\fbp}{{\sf FullyBalancedPartition}\xspace}
\newcommand{\rbp}{{\sf RelaxBalancedPartition}\xspace}

\newcommand{\prob}[1]       {\Pr\left\{ #1 \right\}}

\newcommand{\set}[1]        {\left\{ #1 \right\}}

\newtheorem{theorem}{Theorem}
\newtheorem{corollary}[theorem]{Corollary}
\newtheorem{lemma}[theorem]{Lemma}

\newcommand{\qed}{\rule{7pt}{7pt}}

\newenvironment{proof}{\noindent{\bf Proof}\hspace*{1em}}{\qed\bigskip}
\newenvironment{proof-sketch}{\noindent{\bf Sketch of Proof}\hspace*{1em}}{\qed\bigskip}
\newenvironment{proof-idea}{\noindent{\bf Proof Idea}\hspace*{1em}}{\qed\bigskip}
\newenvironment{proof-of-lemma}[1]{\noindent{\bf Proof of Lemma #1}\hspace*{1em}}{\qed\bigskip}
\newenvironment{proof-of-theorem}[1]{\noindent{\bf Proof of Theorem #1}\hspace*{1em}}{\qed\bigskip}

\newif\ifcodeinbox

\newcounter{codelinenumber}
\newcommand{\zeroli}{\setcounter{codelinenumber}0}

\def\@startline{\global\@curtabmar\@nxttabmar\relax
   \global\@curtab\@curtabmar\setbox\@curline\hbox
    {}\@startfield\global\lifalse\strut}

\newenvironment{code}{\global\codeinboxtrue%
\setbox\strutbox\hbox{\vrule height 9pt depth 4pt width0pt}%
\noindent\begin{tabbing}%
\zeroli\setlength{\tabbingsep}{1em}
\hspace*{1em}\=999\ifdoubledigit9\fi
\=\ {\bf if} \={\bf then} \={\bf if} \={\bf then}
    \={\bf if} \={\bf then} \={\bf if} \={\bf then} \={\bf if} \={\bf then}
    \={\bf if} \=\+\+\kill}{\end{tabbing}\global\codeinboxfalse}

\newif\ifdoubledigit
\newcommand{\codebox}[1]{\setbox0=\vbox{\begin{code}#1\end{code}}
\ifnum\c@codelinenumber>9\global\doubledigittrue\else\doubledigitfalse\fi
\vskip1sp\noindent\hskip-14pt
\parbox{\textwidth}{\begin{code}\protect#1\end{code}}\global\let\@currentlabel=\thechapter}

\newif\ifli
\newcommand{\li}{\global\litrue\stepcounter{codelinenumber}%
\ifdoubledigit%
\hbox to8pt{\hss\thecodelinenumber\hskip5pt}\else%
\hbox to8pt{\hskip-1pt\thecodelinenumber\hss}\fi%
\xdef\@currentlabel{\p@codelinenumber\thecodelinenumber}\'}

\newcommand{\secput}[2]{\section{#2}\label{sec:#1}}

\def\Reals#1{\mathbb{R}^{#1}}

\def\sizeof#1{\left|#1  \right|}

\newcommand{\secref}[1]     {Section~\ref{sec:#1}}

\newcommand{\thmlabel}[1]   {\label{thm:#1}}
\newcommand{\thmref}[1]     {Theorem~\ref{thm:#1}}
\newcommand{\lemlabel}[1]   {\label{lem:#1}}
\newcommand{\lemref}[1]     {Lemma~\ref{lem:#1}}
\newcommand{\figref}[1]     {Figure~\ref{fig:#1}}
\newcommand{\figlabel}[1]   {\label{fig:#1}}

\newcommand{\eqlabel}[1]     {\label{eq:#1}}
\renewcommand{\eqref}[1]      {(\ref{eq:#1})}

\def\hyperspc{\kern -0.22em}

\newenvironment{fminipage}%
  {\begin{Sbox}\begin{minipage}}%
  {\end{minipage}\end{Sbox}\fbox{\TheSbox}}

\newenvironment{kbalgbox}[0]{\vskip 0.0in
\noindent
\begin{fminipage}{0.99\linewidth}
}{
\end{fminipage}
\vskip 0.1in
}

\newcounter{@inst}
\newcounter{@auth}

\newdimen\instindent
\newbox\authrun
\newtoks\authorrunning
\newtoks\tocauthor
\newbox\titrun
\newtoks\titlerunning
\newtoks\toctitle

\def\institute#1{\gdef\@institute{#1}}
\def\inst#1{\unskip$^{}$}

\newcommand{\defn}[1]           {{\textit{\textbf{\boldmath #1}}}}

\newcommand{\bigO}[1]{O({#1})}

\newcommand{\crossing}{\mbox{\sf crossing}\xspace}
\newcommand{\inner}{\mbox{\sf inner}\xspace}
\newcommand{\outgoing}{\mbox{\sf outgoing}\xspace}
\renewcommand{\outer}{\mbox{\sf outer}\xspace}

\newcommand{\graphsize}[1]{|#1|}
\newcommand{\ang}[1]        {\ifmmode{\left\langle #1 \right\rangle}
                 \else{$\left\langle${#1}$\right\rangle$}\fi}

\newcommand{\tree}{T}
\newcommand{\fbt}{BT}

\newcommand{\flb}{relax-balanced\xspace}

\newcommand{\FLB}{Relax-Balanced\xspace}

\newcommand{\rb}{\flb}

\newcommand{\RB}{\FLB}

\begin{document}

\title{Optimal Cache-Oblivious Mesh Layouts}

\author{Michael A. Bender
    \thanks{\sloppy Department of Computer Science, Stony Brook University,
            Stony Brook, NY 11794-4400, USA.
            Email:~\texttt{bender@cs.sunysb.edu}.}~
    \thanks{Tokutek, Inc. \texttt{http://www.tokutek.com}.}~
    \thanks{\sloppy Supported in part by NSF Grants CCF~0621439/0621425,
            CCF~0540897/05414009, CCF 0634793/0632838, CNS 0627645, and
            CCF~0937822 and by 
            DOE Grant DE-FG02-08ER25853.}
\and Bradley~C.~Kuszmaul
    \thanks{\sloppy Computer Science and Artificial Intelligence Laboratory,
            Massachusetts Institute of Technology,
            Cambridge, MA 02139, USA.
            Email:~\texttt{bradley@mit.edu}}~
            \footnotemark[2]~
    \thanks{\sloppy Supported in part by the
     Singapore-MIT Alliance, NSF Grant ACI-0324974, and DOE Grant DE-FG02-08ER25853.}
\and
  Shang-Hua Teng
    \thanks{\sloppy Computer Science Department,
  University of Southern California,
  941 Bloom Walk,
  Los Angeles, CA 90089-0781, USA. 
            Email:~\texttt{shanghua@usc.edu.}}~
    \thanks{Akamai Technologies, Inc.}~
    \thanks{\sloppy Supported in part by NSF grants
            CCR-0311430 and ITR CCR-0325630.}
\and Kebin Wang
	    \thanks{\sloppy Computer Science Department,
            Boston University, Boston, MA 02215, USA.
            Email:~\texttt{kwang@cs.bu.edu}.}
}


\maketitle


\begin{abstract}

  A \defn{mesh} is a graph that divides physical space into
  regularly-shaped regions.  Meshes computations form the basis of many
  applications, including finite-element methods, image rendering,
  collision detection, and N-body simulations.  In one important mesh
  primitive, called a \defn{mesh update}, each mesh vertex stores a
  value and repeatedly updates this value based on the values stored in all
  neighboring vertices.  The performance of a mesh update depends on the
  layout of the mesh in memory. Informally, if the mesh layout has
  good data locality (most edges connect a pair of nodes that are stored
  near each other in memory), then a mesh update runs quickly.

  This paper shows how to find a memory layout that guarantees that the
  mesh update has asymptotically optimal memory performance for any
  set of memory parameters. Specifically, the cost of the mesh update
  is roughly the cost of a sequential memory scan. Such a memory
  layout is called \defn{cache-oblivious}. Formally, for a
  $d$-dimensional mesh $G$, block size $B$, and cache size $M$ (where
  $M=\Omega(B^d)$), the mesh update of $G$ uses $O(1+\sizeof{G}/B)$
  memory transfers.  The paper also shows how the mesh-update performance
  degrades for smaller caches, where $M=o(B^d)$.

  The paper then gives two algorithms for finding cache-oblivious mesh
  layouts.  The first layout algorithm runs in time
  $O(\sizeof{G}\log^2\sizeof{G})$ both in expectation and with high
  probability on a RAM. 
  It uses $O(1+\sizeof{G}\log^2(\sizeof{G}/M)/B)$ memory transfers in expectation
  and $O(1+(\sizeof{G}/B)(\log^2(\sizeof{G}/M) + \log\sizeof{G}))$ memory transfers with high
  probability in the cache-oblivious and disk-access machine (DAM)
  models. The layout is obtained by finding a fully balanced
  decomposition tree of $G$ and then performing an in-order traversal of
  the leaves of the tree.

  The second algorithm computes a cache-oblivious layout on a RAM in time
  $O(\sizeof{G}\log\sizeof{G}\log\log\sizeof{G})$ 
  both in expectation and with high probability.
  In the DAM and cache-oblivious models, the second layout algorithm uses
  $\bigO{1+(\graphsize{G}/B)\log{(\graphsize{G}/M)}\min\{\log\log{\sizeof{G}}, \log(\sizeof{G}/M)\}}$
  memory transfers in expectation and 
  $\bigO{1+(\graphsize{G}/B)(\log{(\graphsize{G}/M)} \min\{\log\log{\sizeof{G}}, \log(\sizeof{G}/M)\}+ \log{\sizeof{G}})}$
  memory transfers with high probability.
  The
  algorithm is based on a new type of decomposition tree, here called
  a \defn{\flb decomposition tree}. Again, the layout is obtained by
  performing an in-order traversal of the leaves of the decomposition
  tree.

\end{abstract}




\secput{intro}{Introduction}

A \defn{mesh} is a graph that represents a division of physical space
into regions, called \defn{simplices}. Simplices are typically
triangular (in 2D) or tetrahedral (in 3D). They are \defn{well shaped},
which informally means that they cannot be long and skinny, but must
be roughly the same size in any direction.  Meshes form the basis of
many computations such as finite-element methods, image rendering,
collision detection, and N-body simulations. Constant-dimension meshes
have nodes of constant-degree.

In one important mesh primitive, each mesh vertex
stores a value and repeatedly updates this value based on the values stored
in all neighboring vertices. Thus, we view the mesh as a weighted graph
$G=(V,E,w,e)$ ($w:V\rightarrow\Reals{}$, $e:E\rightarrow\Reals{+}$).
For each vertex $i\in V$, we repeatedly recompute its weight $w_i$ as
follows:
\[ w_i = \sum_{(i,j)\in E} w_j \, e_{ij}\,. \]
We call this primitive
a \defn{mesh update}. Expressed differently, a mesh update is the
sparse matrix-vector multiplication, where the matrix is the
(weighted) adjacency matrix of $G$, and vectors are the vertex weights.

On a \defn{random access machine (RAM)} (a flat memory model), a mesh
update runs in linear time, regardless of how the data is laid out in
memory.  In contrast, on a modern computer with a hierarchical memory,
how the mesh is laid out in memory can
affect the speed of the computation substantially.  This paper studies the
\defn{mesh layout problem}, which is how to lay out a mesh in memory, so
that mesh updates run rapidly on a hierarchical memory.

We analyze the mesh layout problem in the \defn{disk-access
machine (DAM) model}~\cite{Aggarwal-Vitter-1988} (also known as the
\defn{I/O-model}) and in the \defn{cache-oblivious (CO)
model}~\cite{FrigoLePrRa99}. The DAM model is an idealized
two-level memory hierarchy. These two levels could represent L2
cache and main memory, main memory and disk, or any other pair of
levels. The small level (herein called \defn{cache}) has size $M$,
and the large level (herein called \defn{disk}) has unbounded
size. Data is transferred between the two levels in blocks of size
$B$; we call these \defn{memory transfers}. Thus, a  memory
transfer is a cache-miss if the DAM represents L2 cache and main
memory and is a page fault, if the DAM represents main memory and
disk.

A memory transfer has unit cost. The objective is to minimize the
number of memory transfers.  Focusing on memory transfers, to the
exclusion of other computation, frequently provides a good model
of the running time of an algorithm on a modern computer.  The
\defn{cache-oblivious model} is essentially the DAM model, except that
the values of $B$ and $M$ are unknown to the algorithm or the
coder.  The main idea of cache-obliviousness is this: If an
algorithm performs an asymptotically optimal number of memory
transfers on a DAM, but the algorithm is not parameterized
by $B$ and $M$, then the algorithm also performs an asymptotically
optimal number of memory transfers on an arbitrary unknown,
multilevel memory hierarchy.

The cost of a mesh update in the DAM and
cache-oblivious models depends on how the mesh is laid out in
memory. An update to a mesh $G = (V,E)$ is just a graph traversal.
If we store $G$'s vertices arbitrarily in memory, then the update
could cost as much as $\bigO{|V|+|E|} = \bigO{\sizeof{G}}$ memory
transfers, one transfer for each vertex and each edge.  In this paper
we achieve only $\Theta\left(1+|G|/B\right)$ memory transfers.
This is the cost of a sequential scan of a chunk of memory
of size $O(\sizeof{G})$, which is asymptotically optimal.

Our mesh layout algorithms extend earlier ideas from VLSI theory.
Classical VLSI-layout algorithms turn out to have direct
application in scientific and I/O-efficient computing. Although
these diverse areas may appear unrelated, there are important
parallels.  For example, in a good mesh layout, vertices are
stored in (one-dimensional) \emph{memory locations} so that most
mesh edges are short; in a good VLSI layout, graph vertices are
assigned to (two-dimensional) \emph{chip locations} so that most
edges are short (to cover minimal area).

\subsection*{Results}
We give two algorithms for laying out a constant-dimension well-shaped
mesh $G=(V,E)$ so that updates run in $\Theta(1+\sizeof{G}/B)$ memory
transfers, which is $\Theta(1+\sizeof{V}/B)$ since the mesh has
constant degree.

\sloppy Our first layout algorithm runs in time
$O(\sizeof{G}\log^2 \sizeof{G})$ on a RAM both in expectation and
with high probability.\footnote{For input size $N$ and event $E$,
we say that $E$ occurs \defn{with high probability} if for any
constant $c>0$ there exists a proper choice of constants defining
the event such that $\prob{E}\geq 1-N^{-c}$.}
In the DAM and cache-oblivious models,
the algorithm uses
$\bigO{1+(\graphsize{G}/B)\log^2{(\graphsize{G}/M)}}$ memory transfers
in expectation and
$\bigO{1+(\graphsize{G}/B) (\log^2{(\graphsize{G}/M)} +\log{\graphsize{G}}) }$
memory transfers with high probability.
The layout algorithm is based on decomposition
trees and fully balanced decomposition
trees~\cite{Leighton82,BhattLe84}; specifically, our mesh layout is
obtained by performing an in-order traversal of the leaves of a
fully-balanced decomposition tree.  Decomposition trees were developed
several decades ago as a framework for VLSI 
layout~\cite{Leighton82,BhattLe84}, but they are well suited for mesh
layout. However, the original algorithm for building fully-balanced
decomposition trees is too slow for our uses (it appears to run in
time $\bigO{\sizeof{G}^{\Theta(b)}}$, where $b$ is the degree bound of
the mesh). Here we develop a new algorithm that is faster and simpler.

\sloppy Our second layout algorithm, this paper's main result,
runs in time $O(\sizeof{G}\log \sizeof{G}\log\log \sizeof{G})$ on
a RAM both in expectation and with high probability. In the DAM
and cache-oblivious models, the algorithm uses
$O(1+(\sizeof{G}/B)\log(\sizeof{G}/M) \min\{\log\log{\sizeof{G}},
  \log(\sizeof{G}/M)\}
)$ memory transfers in expectation and
$O(1+(\sizeof{G}/B)(\log(\sizeof{G}/M)\min\{\log\log{\sizeof{G}},
  \log(\sizeof{G}/M)\}+\log\sizeof{G}))$
memory transfers with high probability.

The algorithm is based on a new type of
decomposition tree, which we call a \defn{\flb decomposition tree}.
As before, our mesh layout is obtained by performing an in-order
traversal of the leaves of a \flb decomposition tree. By carefully
relaxing the requirements of decomposition trees, we can retain
asymptotically optimal mesh updates, while improving construction by
nearly a logarithmic factor.

The mesh-update guarantees require a \defn{tall-cache assumption} on
the memory system that $M=\Omega(B^d)$, where $d$ is the dimension of
the mesh. We also show how the performance degrades for small caches,
where $M=o(B^d)$. If the cache only has size $\bigO{B^{d-\epsilon}}$,
then the number of memory transfers increases to
$\bigO{1+\sizeof{G}/B^{1-\epsilon/d}}$.

In addition to the main results listed above, this paper has
contributions extending beyond I/O-efficient computing.  First,
our algorithms for building fully-balanced decomposition trees are
faster and simpler than previously known algorithms.  Second, our
\flb decomposition trees may permit some existing algorithms based
on decomposition trees to run more quickly.  Third, the techniques
in this paper yield simpler and improved methods for generating
$k$-way partitions of meshes, earlier shown
in~\cite{KiwiSpielmanTeng}.  More generally, we cross-pollinate
several fields, including I/O-efficient computing, VLSI layout,
and scientific computing.




\secput{balanced}{Geometric Separators and Decomposition Trees}

In this section we review the geometric-separator
theorem~\cite{MillerTengThurstonVavasis}, which we use for
partitioning constant-dimensional meshes. We then review decomposition
trees~\cite{Leighton82}. Finally, we show how to use geometric
separators to build decomposition trees for well shaped meshes.

\subsection*{Geometric Separators}

A finite-element mesh is a decomposition of a geometric domain into a
collection of interior-disjoint \defn{simplices} (e.g., triangles in
2D and tetrahedra in 3D), so that two simplices can only intersect at
a lower dimensional simplex. Each simplicial element of the mesh must
be \defn{well shaped}.  Well shaped means that there is a constant
upper bound to the aspect ratio, that is, the ratio of the radius of
the smallest ball containing the element to the radius of the largest
ball contained in the element~\cite{TengWong}.

A \defn{partition} of a graph $G=(V,E)$ is a division of $G$ into
disjoint subgraphs $G_0=(V_0,E_0)$ and $G_1=(V_1,E_1)$ such that
$V_0\cap V_1=\emptyset$, and $V_0\cup V_1 =V$.  $G_0$ and $G_1$ is a
\defn{$\beta$-partition} of $G$ if they are a partition of $G$ and
$\sizeof{V_0}, \sizeof{V_1} \leq \beta \sizeof{V}$.  We let
$E(G_0,G_1)$ denote the set of edges in $G$ crossing from $V_0$ to
$V_1$, and $E(v,G_1)$ denote the set of edges in $G$ connecting vertex
$v$ to the vertices of $G_1$.  For a function $f$, $G=(V,E)$ has a 
\defn{family of $(f,\beta)$-partitions} if for each subset $S \subseteq
V$ and induced graph $G_S =(V_S,E_S) $, 
graph $G_S$ has a $\beta$-partition of
$G_{S_0}=(V_{S_0},E_{S_0})$ and $G_{S_1}=(V_{S_1},E_{S_1})$
such that $ \sizeof{E_S - E_{S_0} -  E_{S_1}} \leq f(\sizeof{V_S})$.

The following separator theorem of Miller, Teng, Thurston, and
Vavasis~\cite{MillerTengThurstonVavasis} shows that
meshes can be partitioned efficiently:

\begin{theorem}[Geometric Separators~\cite{MillerTengThurstonVavasis}]%
Let $G=(V,E)$ be a well shaped finite-element mesh in $d$ dimensions
($d>1$). For constants $\epsilon$ ($0<\epsilon<1$) and $c(\epsilon,
d)$ depending only on $\epsilon$ and $d$, a
$(f(N)=O(N^{1-1/d}),(d+1+\epsilon)/(d+2))$-partition of $G$ can be
computed in $O(d\sizeof{G} +c(\epsilon, d) )$ time with probability
at least $1/2$.
\thmlabel{GeoSep}
\end{theorem}

The separator algorithm from~\cite{MillerTengThurstonVavasis} works as
follows.  First, project the coordinates of the vertices of the input
graph $G$ onto the surface of a unit sphere in $(d+1)$-dimensions.
The projection of each point is independent of all other input points
and takes constant time.  Sample a constant number of points from all
projected points uniformly at random.  Compute a centerpoint of the
sampled points.  (A \defn{centerpoint} of a point set in
$d$-dimensions is a point such that every hyperplane through the
centerpoint divides the point set approximately evenly, i.e., in the
ratio of $d$ to $1$ or better.)  Rotate and then dilate the sampled
points.  Both the rotation and dilation are functions of the
centerpoint and the dimension $d$.  Choose a random great circle on
this unit sphere. (A \defn{great circle} of a sphere is a circle on
the sphere's surface that evenly splits the sphere.)  Map the great
circle back to a sphere in the $d$-dimensional space by reversing
the dilation, the rotation, and the projection.  Now use this new
sphere to divide the vertices and the edges of the input graph.

Now more mechanics of the algorithm.  Mesh $G$ is stored in an
array. Each vertex of $G$ is stored with its index (i.e., name), its
coordinates, and all of its adjacent edges, including the index and
coordinates of all neighboring vertices.  (This mesh representation
means that each edge is stored twice, once for each of the edge's two
vertices.)

To run the algorithm,
scan the vertices and edges in $G$ after obtaining the sphere
separator. During the scan, divide the
vertices into two sets, $G_0$, containing the vertices inside the new
sphere and $G_1$, containing the vertices outside the sphere.  Mark an
edge as ``crossing'' if the edge crosses from $G_{0}$ to $G_{1}$.
Verify that the number of crossing edges, $\sizeof{E(G_0,G_1)}$, is
$O(\sizeof{G}^{1-1/d})$, and if not, repeat. The cost of this
scan is $O(\sizeof{G}/B+1)$ memory transfers.

The geometric separator algorithm has the following performance:
\begin{corollary}%
\label{cor:GeoSep} Let $G=(V,E)$ be a well shaped finite-element
mesh in $d$ dimensions ($d>1$).
For constants $\epsilon$ ($0<\epsilon<1$) and $c(\epsilon,
d)$ depending only on $\epsilon$ and $d$,
the geometric-separator algorithm
finds an $(f(N)=O(N^{1-1/d}),(d+1+\epsilon)/(d+2))$-partition of $G$. The
algorithm runs in $O(\graphsize{G})$ on a RAM and uses
$O(1+\graphsize{G}/B)$ memory transfers in the DAM and
cache-oblivious models, both in expectation and with probability at
least $1/2$. With high probability, the geometric-separator
algorithm completes in $O(\graphsize{G}\log\graphsize{G})$ on a RAM
and uses $O(1+\graphsize{G}\log\graphsize{G}/B)$ memory transfers in
the DAM and cache-oblivious models.
\end{corollary}

\begin{proof}
A linear scan of $G$ takes time $O(\sizeof{G})$ and uses an asymptotically
optimal number of memory transfers.  We
expect to find a good separator after a constant number of trials, and
so the expectation bounds follow by linearity of expectation.
The probability that after selecting $c\lg{\sizeof{G}}$ candidate
separators, none are good is at most
$1/2^{c\lg{\sizeof{G}}} ={\sizeof{G}}^{-c}$.
Thus, with high probability, the geometric separator algorithm
completes in $\bigO{\graphsize{G}\log\graphsize{G}}$ on a RAM and
uses $\bigO{1+\graphsize{G}\log\graphsize{G}/B}$ memory transfers in
the DAM and cache-oblivious models.
The separator algorithm is cache-oblivious since it is not parameterized
by $B$ or $M$.
\end{proof}

\subsection*{Decomposition Trees}

A \defn{decomposition tree} $T_G$ of a graph $G=(V,E)$ is a recursive
partitioning of $G$. The root of $T_G$ is $G$. Root $G$ has left and right children
$G_0$ and $G_1$, and grandchildren $G_{00}$, $G_{01}$, $G_{10}$, $G_{11}$, and so on
recursively down the tree. Graphs $G_{0}$ and $G_{1}$ partition $G$, graphs  $G_{00}$ and $G_{01}$ partition $G_{0}$, and so on. More generally,
a node in the decomposition tree is denoted $G_p$ ($G_p\subset G$), where $p$
is a bit string representing the path to that node from the root.
We call $p$ the \defn{id} of $G_p$.  We
say that a decomposition tree is \defn{$\beta$-balanced} if for all
siblings $G_{p0}=(V_{p0},E_{p0})$ and $G_{p1}=(V_{p1},E_{p1})$ in the tree,
$\sizeof{V_{p0}},\sizeof{V_{p1}}\leq\beta\sizeof{V_p}$. We say that a
decomposition tree is \defn{balanced} if $\beta=1/2$.
For a function
$f$, \defn{graph $G$ has an $f$ decomposition tree} if for all
(nonleaf) nodes $G_{p}$ in the decomposition tree,
$\sizeof{E(G_{p0},G_{p1})}\leq f(\sizeof{V_{p}})$.
A $\beta$-balanced $f$ decomposition tree is abbreviated as an
$(f,\beta)$-decomposition tree.

For a parent node $G_p$ and its children $G_{p0}$ and $G_{p1}$, there are several
categories of edges.
\defn{Inner edges}
connect vertices that are both in $G_{p0}$ or both in $G_{p1}$.
\defn{Crossing edges} connect vertices in $G_{p0}$ to vertices in
$G_{p1}$. \defn{Outgoing edges} of $\boldmath G_{p0}$ (resp.\
$G_{p1}$) connect vertices in $G_{p0}$ (resp.\ $G_{p1}$) to
vertices in neither set, i.e., to vertices in $G-G_{p}$.
\defn{Outer edges} of $\boldmath G_{p0}$ (resp.\
$G_{p1}$) connect vertices in $G_{p0}$ (resp.\ $G_{p1}$)
to vertices in $G-G_{p0}$ (resp.
$G-G_{p1}$); thus an outer edge is either a crossing edge or an outgoing edge.
More formally,
$\inner(G_{p0})=E(G_{p0},G_{p0})$,
$\crossing(G_{p})=E(G_{p0},G_{p1})$,
$\outgoing(G_{p0})=E(G_{p0},G-G_{p})$, and
$\outer(G_{p0})=E(G_{p0},G-G_{p0})$.

We build a decomposition tree  $T_{G}$ of mesh $G$ recursively.
First we run the geometric separator algorithm on the root $G$ to find
the left and right children, $G_0$ and $G_1$.
Then we recursively build the decomposition tree rooted at $G_0$
and then the decomposition tree rooted at $G_1$.
(Thus, the right child of $T_G$ is not processed until the
whole left subtree is built.)

The decomposition tree is encoded as follows. Each leaf node
$G_{q}$ for $T_{G}$ stores the single vertex $v$ and the bit
string $q$ (the root-to-leaf path).
The leaf nodes of $T_G$ are stored contiguously in an array $L_G$.
The bit string $q$ contains enough information to
determine which nodes (subgraphs) of $T_{G}$ contain $v$ --- specifically
any node $G_{\hat{q}}$, where $\hat{q}$ is a prefix of $q$
(including $q$).  As mentioned earlier,
each vertex is stored along with its coordinates, adjacent edges, and coordinates  of all neighboring vertices in $G$.
(Recall that  each edge is therefore stored twice, once for each
of the edge's vertices.)
Each edge $e$ in $G$ is a crossing edge for exactly one node
in the decomposition tree $T_{G}$.  In $T_{G}$, each edge $e$ also
stores the id $p$ of the tree node $G_{p}$ for which $e$ is a
crossing edge.  The bit strings on nodes and edges therefore
contains enough information to
determine which edges are crossing, inner, and outer for which
tree nodes. Specifically, $e\in\crossing(G_{p})$. Let $\hat{p}$ be
a prefix of $p$ that is strictly shorter ($p\neq\hat{p}$);
then
$e\in\inner(G_{\hat{p}})$.
Let $\tilde{p}$
be bit string representing a node in $T_{G}$
where $p$ is a strictly shorter prefix of $\tilde{p}$ ($p \neq \tilde{p}$).
Then $e\in\outer(G_{\tilde{p}})$. If $\tilde{p}0$ and $\tilde{p}1$
represent nodes in $T_{G}$, then
$e\in\outgoing(G_{\tilde{p}0})$ or
$e\in\outgoing(G_{\tilde{p}1})$.

Thus, decomposition tree $T_{G}$ is laid out in memory by storing the
leaves in order in an array $L_{G}$. We do not need to store internal
nodes explicitly because the bit strings on nodes and edges encode
the tree structure.

Here are a few facts about our layout of $T_{G}$.
Given any two nodes $G_{p}$ and $G_{q}$ of $L_G$,
the common prefix of $p$ and $q$ is the smallest node in $T_{G}$
containing all vertices in both $G_{p}$ and $G_{q}$.
All the vertices in any node $G_{p}$ of $T_{G}$ are stored in a
single contiguous chunk of the array.
Thus, we can identify for $G_{p}$, which edges are inner, crossing,
outer, and outgoing by performing a single linear scan
of size $O(\sizeof{G_{p}})$.

We construct the decomposition tree $T_{G}$ by recursively partitioning of
$G$. While $T_{G}$ is in the process of being constructed, its
encoding is similar to the above, except that (1) a leaf node
$G_{q}$ may contain more than one vertex, and (2) some edges may
not yet be labelled as crossing. Thus, when the process begins,
$T_{G}$ is just a single leaf comprising $G$. The nodes are stored
in a single array of size $O(\sizeof{G})$ and are stored in an
arbitrary order.  Then we run the geometric separator algorithm.
Once we find a good separator, we partition $G$ into $G_{0}$ and
$G_{1}$, and we store $G_{0}$ before $G_{1}$ in the same array. We
label vertices of $G_0$ with bit string $0$ and vertices of $G_1$
with bit string $1$. We then run through and label all crossing
edges with the appropriate bit string (for the leaf node, the
empty string). Now the nodes in each of $G_{0}$ and $G_{1}$ are
stored in an arbitrary order, but the subarray containing $G_{0}$
is stored before the subarray containing $G_{1}$.  We then apply
the geometric separator algorithm for $G_{0}$. We partition into
$G_{00}$ and $G_{01}$, label vertices in $G_{0}$ with $00$ or
$01$, and label all crossing edges of $G_0$ with the bit string
$0$; we then do the same for $G_{00}$ and so on recursively until
all leaf nodes are graphs containing a single vertex.

We now give the complexity of building the decomposition tree.
Our high-probability bounds are based on the following observation
involving a coin with a constant probability of heads.
In order to get at least one head with
probability at least $1-1/\poly(N)$, $\Theta(\log N)$
flips are necessary and sufficient.
In order to get $\Theta(\log N)$ heads with
probability at least $1-1/\poly(N)$,
the asymptotics do not change; $\Theta(\log N)$ flips are still
necessary and sufficient.
The following lemma can be proved by Chernoff bounds (or otherwise):

\begin{lemma}
Consider $S\geq c\log N$ flips of a coin with a constant probability of
heads, for sufficiently large constant $c$.
With probability at least $1-1/\poly(N)$,
$\Theta(S)$ of the flips are heads.
\lemlabel{lognlemma}
\end{lemma}

\begin{theorem}%
Let $G=(V,E)$ be a well shaped finite-element
mesh in $d$ dimensions ($d>1$).  Mesh $G$ has a
$(2d+3)/(2d+4)$-balanced-$\bigO{{\sizeof{V}^{1-1/d}}}$
decomposition tree. On a RAM, the decomposition tree can be computed in time
$\bigO{\graphsize{G}\log{\graphsize{G}}}$ both in expectation and
with high probability. The decomposition tree can
be computed in the DAM and cache-oblivious models using
$\bigO{1+(\graphsize{G}/B)\log{(\graphsize{G}/M)}}$ memory transfers
in expectation and
$\bigO{1+(\graphsize{G}/B)\log\graphsize{G}}$ memory transfers
with high probability.
\thmlabel{mesh-decomposition-tree}
\end{theorem}

\begin{proof}
We first establish that the tree construction takes time
$\bigO{\graphsize{G}\log\graphsize{G}}$  on a RAM in
expectation.
The height of the decomposition tree is $\bigO{\log\graphsize{G}}$, and
the total size of all subgraphs at each height is
$\bigO{\graphsize{G}}$. Since the decomposition of a subgraph takes expected
linear time, the time bounds follow by linearity of expectation.

We next establish that the tree construction
uses $\bigO{1+(\graphsize{G}/B)\log(\graphsize{G}/M)}$ expected
memory transfers in the DAM and cache-oblivious models.
Because we build the decomposition tree recursively, we give a recursive 
analysis. The base case is when
a subtree first has size less than $M$.
For the base case, the cost to build the entire subtree
is $O(M/B)$ because this is the cost to read all
blocks of the subtree into memory.  Said differently,
once a subgraph is a constant fraction smaller than $M$, the cost to build the
decomposition tree from the subgraph is $0$, because all necessary memory blocks already reside in memory. For the recursive step, 
recall that when a subgraph $G_{p}$ has size
greater than $M$, the decomposition of a subgraph takes expected
$O(\sizeof{G_p}/B)$ memory transfers, because this is the cost of a linear scan.
Thus, there are
$\bigO{\log{(\graphsize{G}/M)}}$ levels of the tree with subgraphs bigger than 
$M$, so the algorithms uses expected 
$\bigO{1+(\graphsize{G}/B)\log{(\graphsize{G}/M)}}$ memory transfers.

We next establish the high-probability bounds.
We show that the building process uses
$\bigO{\graphsize{G}\log{\graphsize{G}}}$ time on a RAM
and $\bigO{1+\graphsize{G}\log\graphsize{G}/B}$ memory transfers
in the DAM and the cache-oblivious models with high probability.

First consider all nodes that have size
$\Omega(\sizeof{G}/\log\sizeof{G})$.
There are $\Theta(\log{\sizeof{G}})$ such nodes.
To build these nodes, we require a total of
$\Theta(\log{\sizeof{G}})$ good separators. We can view finding these
separators as a coin-flipping game, where we need $\Theta(\log{\sizeof{G}})$
heads;  by~\lemref{lognlemma} we require
$\Theta(\log{\sizeof{G}})$ coin flips.
However, separators near the top of the tree
are more expensive to find than separators deeper in the tree.
We bound the cost to find all of these separators by the cost to build the root
separator.
Thus, building these nodes uses time
$O(\sizeof{G}\log{\sizeof{G}})$ and $O(1+\sizeof{G}\log{\sizeof{G}}/B)$
memory transfers with high probability.
This is now the dominant term in the cost to build the decomposition tree.

Further down the tree, where nodes have size $O(\sizeof{G}/\log\sizeof{G})$,
the analysis is easier.
Divide the nodes to be
partitioned into groups whose sizes are within a constant
factor of each other. Now each group contains $\Omega(\log\sizeof{G})$ elements.
Thus, by~\lemref{lognlemma} the time  to build
the rest of the tree  with high probability
equals the time in expectation, which is $\Theta(\sizeof{G}\log{\sizeof{G}})$.

We now finish the bound on the number of memory transfers. As above,
because we build the decomposition tree recursively,
subtrees a constant fraction smaller than $M$ are build for free.
Also, because each group contains $\Omega(\log\sizeof{G})$ elements,
the cost to build these lower levels in the tree with high probability equals
the expected cost, which is
$O(1+(\sizeof{G}/B)\log{(\sizeof{G}/M)})$.
This cost is dominated by the cost to
build the nodes of size $\Omega(\sizeof{G}/\log\sizeof{G})$.
\end{proof}




\secput{fully-balanced}{Fully-Balanced Decomposition Trees for
Meshes}

In this section we define fully-balanced partitions and
fully-balanced decomposition trees.  We give algorithms for
generating these structures on a well shaped mesh $G$. As we show
in \secref{layout}, we use a fully-balanced decomposition tree of
a mesh $G$ to generate a cache-oblivious mesh layout of $G$. Our
construction algorithm is an improvement
over~\cite{Leighton82,BhattLe84} in two respects. First the
algorithm is faster, requiring only $O(\sizeof{G}\log^2
\sizeof{G})$ operations in expectation and with high probability,
$O(1+(\sizeof{G}/B)\log^2(\sizeof{G}/M))$ memory transfers in
expectation, and $\bigO{1+(\graphsize{G}/B)
(\log^2{(\graphsize{G}/M)} +\log{\graphsize{G}}) }$ memory
transfers with high probability. Second, the result is simplified,
no longer relying on a complicated theorem of~\cite{GoldbergWest}.

This section makes it easier to present the
main result of the paper, which appears in \secref{improved}.

\subsection*{Fully-Balanced Partitions}

To begin, we define a fully-balanced partition of a subgraph $G_p$ of $G$. A
\defn{fully-balanced $f$-partition of $G_p\subseteq G$} is a
partitioning of $G_p=(V_p,E_p)$ into two subgraphs
$G_{p0}=(V_{p0},E_{p0})$ and $G_{p1}=(V_{p1},E_{p1})$ such that
\begin{itemize}
    \item $\sizeof{\crossing(G_{p})}\leq f(\sizeof{V_{p}})$,
    \item $\sizeof{V_{p0}} = \sizeof{V_{p1}} \pm O(1)$, and
    \item $\sizeof{\outgoing(G_{p0})}=\sizeof{\outgoing(G_{p1})}\pm O(1)$.%
\end{itemize}

We give the following result before presenting  our algorithm for
computing fully-balanced partitions. The existence proof and time
complexity comprise the easiest case in~\cite{GoldbergWest}.

\begin{lemma}%
\lemlabel{NecklaceBisection}
Given an array $L$ of $N$ elements, where each element is marked either
blue or red, there
exists a subarray that contains half of the blue elements to
within one and half of red elements to within one. Such a subarray
can be found in $O(N)$ time and $O(1+N/B)$ memory transfers cache-obliviously.
\end{lemma}

\begin{proof}
This result is frequently described in terms of ``necklaces.''
Conceptually, attach the two ends of the array together
to make a necklace.
By a simple continuity argument (the easiest case of that in
\cite{GoldbergWest}), the necklace can be split into two
pieces, $A$ and $\bar{A}$, using two cuts such that both pieces
have the same number of blue elements to within one and the same
number of red elements to within one. (For details of the continuity argument,
see~\figref{necklace-fig}.) Translating back to the array,
at least one of $A$ and $\bar{A}$
does not contain the connecting point and is contiguous.

To find a good subarray, first scan $L$ to count the
number of blue elements and the number of red elements. Now rescan $L$, maintaining
a window of size $N/2$. The window initially
contains the first half of $L$ and at the end contains the second half of $L$. (For odd $N$,
the middle element of the array appears in all windows.)
Stop the scan once the window has
the desired number of red and blue elements.

Since only linear scans are used, the algorithm is cache-oblivious
and requires $\Theta(1+N/B)$ memory transfers.
\end{proof}


\begin{figure}[t]
\begin{center}
\begin{picture}(0,0)%
\epsfig{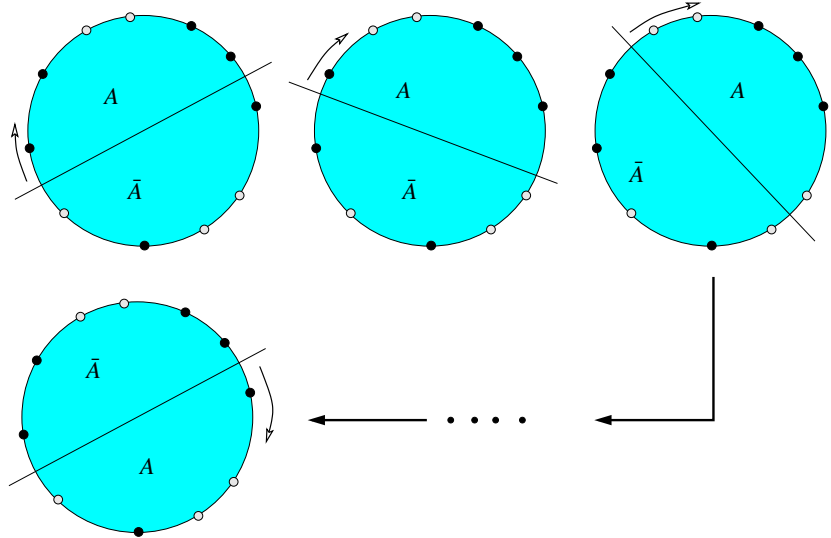}%
\end{picture}%
\setlength{\unitlength}{1973sp}%
\begingroup\makeatletter\ifx\SetFigFont\undefined%
\gdef\SetFigFont#1#2#3#4#5{%
  \reset@font\fontsize{#1}{#2pt}%
  \fontfamily{#3}\fontseries{#4}\fontshape{#5}%
  \selectfont}%
\fi\endgroup%
\begin{picture}(10306,6741)(1339,-7230)
\put(2551,-1786){\makebox(0,0)[lb]{\smash{{\SetFigFont{9}{10.8}{\familydefault}{\mddefault}{\updefault}{\color[rgb]{0,0,0}$A$}%
}}}}
\put(6226,-1711){\makebox(0,0)[lb]{\smash{{\SetFigFont{9}{10.8}{\familydefault}{\mddefault}{\updefault}{\color[rgb]{0,0,0}$A$}%
}}}}
\put(10426,-1711){\makebox(0,0)[lb]{\smash{{\SetFigFont{9}{10.8}{\familydefault}{\mddefault}{\updefault}{\color[rgb]{0,0,0}$A$}%
}}}}
\put(3001,-6436){\makebox(0,0)[lb]{\smash{{\SetFigFont{9}{10.8}{\familydefault}{\mddefault}{\updefault}{\color[rgb]{0,0,0}$A$}%
}}}}
\put(2851,-2986){\makebox(0,0)[lb]{\smash{{\SetFigFont{9}{10.8}{\familydefault}{\mddefault}{\updefault}{\color[rgb]{0,0,0}$\bar{A}$}%
}}}}
\put(6301,-2986){\makebox(0,0)[lb]{\smash{{\SetFigFont{9}{10.8}{\familydefault}{\mddefault}{\updefault}{\color[rgb]{0,0,0}$\bar{A}$}%
}}}}
\put(9151,-2761){\makebox(0,0)[lb]{\smash{{\SetFigFont{9}{10.8}{\familydefault}{\mddefault}{\updefault}{\color[rgb]{0,0,0}$\bar{A}$}%
}}}}
\put(2326,-5236){\makebox(0,0)[lb]{\smash{{\SetFigFont{9}{10.8}{\familydefault}{\mddefault}{\updefault}{\color[rgb]{0,0,0}$\bar{A}$}%
}}}}
\end{picture}%

\end{center}
\caption{Unfilled beads represent blue elements and filled beads
  represent red elements. Pick an arbitrary initial bisection $A$ and
  $\bar{A}$ of the necklace. Here $A$ contains more than half of all
  blue beads. (We can focus exclusively on blue beads because if $A$
  contains half of the blue beads to within one, it also contains half
  of red beads to within one.)  We ``turn'' the bisection clockwise so
  that $A$ takes one bead from $\bar{A}$ and relinquishes one bead to
  $\bar{A}$. Thus, the number of blue beads in $A$ can
  increase/decrease by one or remain the same after each
  turn. However, after $N/2$ turns, $A$ becomes $\bar{A}$, which
  contains less than half of all blue beads. So by a continuity
  argument, $A$ contains half of all blue beads after some number of
  turns. The argument is similar for both odd and even $N$.}
\figlabel{necklace-fig}
\end{figure}


We now present an algorithm for computing fully-balanced partitions.
Given $G_p\subseteq G$, and a $(f(N)=O(N^\alpha), \beta)$-partitioning
geometric separator, \fbp$(G_p)$ computes a fully-balanced
$(f(N)=O(N^\alpha))$-partition $G_{px}$ and $G_{py}$ of $G_p$.

\bigskip

\begin{kbalgbox}%
\label{alg:fully-balanced-partition}
{\fbp}($G_p$)
\begin{enumerate}

\item \emph{Build a decomposition tree} --- Build a decomposition tree
  $T_{G_p}$ of $G_p$ using the $(f(N)=O(N^\alpha),
  \beta)$-partitioning geometric separator.

\item \emph{Build a red-blue array} --- Build an array of blue and red
  elements based on the decomposition tree $T_{G_p}$. Put a blue
  element for each leaf $G_{q}$ in $T_{G_p}$; thus there is a blue
  element for each vertex $v$ in $G_p$. Now insert some red elements
  after each blue element. Specifically, after the blue element
  representing vertex $v$, insert $E(v,G-G_{p})$ red elements. Thus,
  the blue elements represent vertices in $G_{p}=(V_{p},E_{p})$ for a
  total of $\sizeof{V_{p}}$ blue elements, while the red elements
  represent edges to vertices in $G-G_{p}$, for a total of
  $E(G_{p},G-G_{p})$ red elements.

\item \emph{Find a subarray in the red-blue array} --- Find a subarray
  of the red-blue array based on \lemref{NecklaceBisection}. Now
  partition the vertices in $G_{p}$ based on this
  subarray. Specifically, put the vertices representing blue elements
  in the subarray in set $V_{px}$ and put the remaining vertices in
  $G_{p}$ in set $V_{py}$.

\item \emph{Partition $G_{p}$} --- Compute $G_{px}$ and $G_{py}$ from
  $V_{px}$ and $V_{py}$. This computation also means scanning edges to
  determine which edges are internal to $G_{px}$ and $G_{py}$ and
  which have now become external.
\end{enumerate}
\end{kbalgbox}


\begin{figure}[t!]
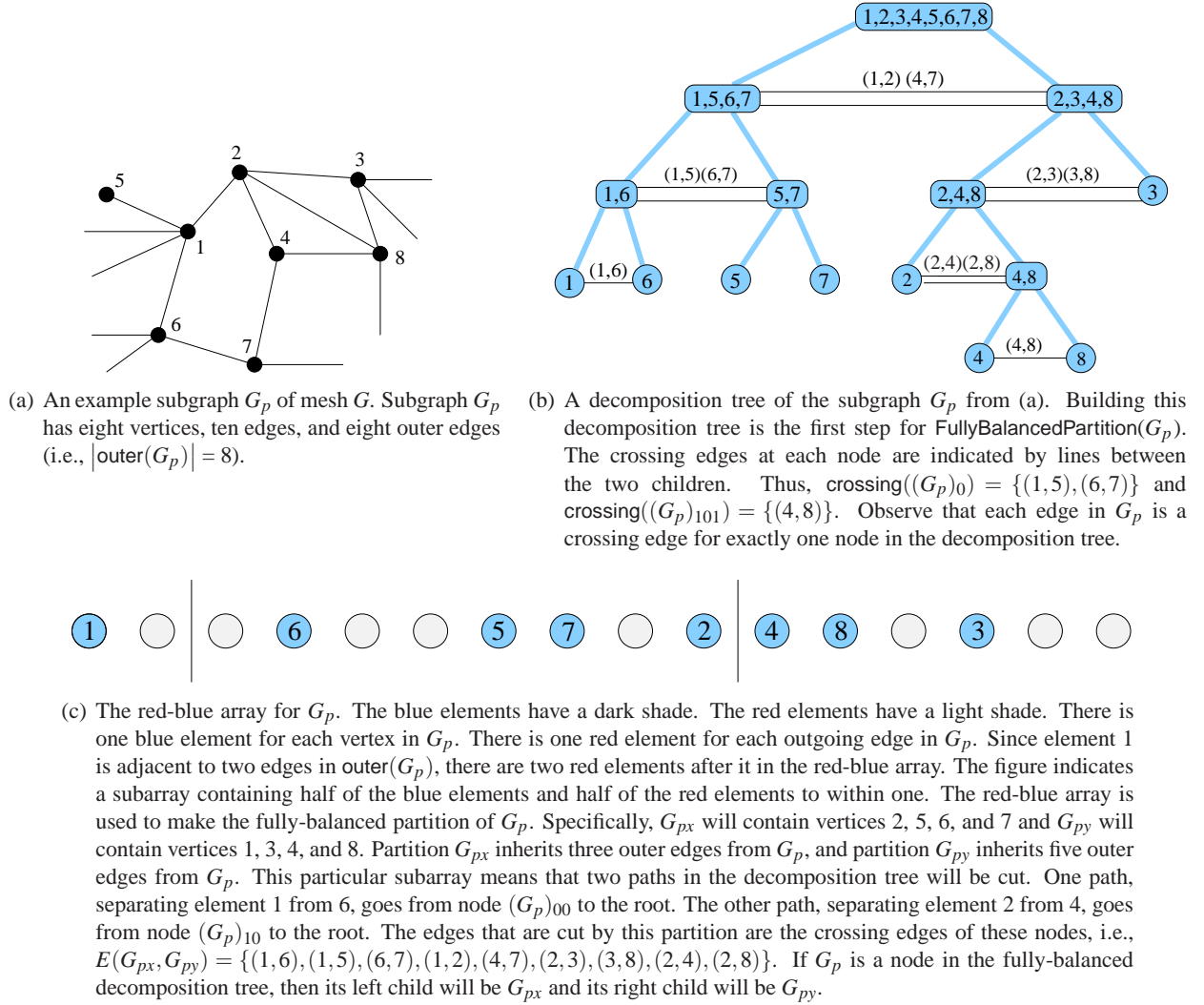

\begin{center}
\subfigure[\sloppy An example subgraph $G_p$ of mesh $G$.
Subgraph $G_p$ has eight vertices,  ten edges, and eight outer edges
(i.e., $\sizeof{\outer(G_p) }$ = 8).
]{
\figlabel{samplegraph} 
\begin{minipage}[b]{0.41\textwidth}
\centering
\begin{picture}(0,0)%
\epsfig{file=figures/samplegraph.pstex}%
\end{picture}%
\setlength{\unitlength}{2565sp}%
\begingroup\makeatletter\ifx\SetFigFont\undefined%
\gdef\SetFigFont#1#2#3#4#5{%
  \reset@font\fontsize{#1}{#2pt}%
  \fontfamily{#3}\fontseries{#4}\fontshape{#5}%
  \selectfont}%
\fi\endgroup%
\begin{picture}(3549,2307)(439,-3523)
\put(1576,-2311){\makebox(0,0)[lb]{\smash{{\SetFigFont{8}{9.6}{\rmdefault}{\mddefault}{\updefault}{\color[rgb]{0,0,0}1}%
}}}}
\put(1329,-3091){\makebox(0,0)[lb]{\smash{{\SetFigFont{8}{9.6}{\rmdefault}{\mddefault}{\updefault}{\color[rgb]{0,0,0}6}%
}}}}
\put(2049,-3316){\makebox(0,0)[lb]{\smash{{\SetFigFont{8}{9.6}{\rmdefault}{\mddefault}{\updefault}{\color[rgb]{0,0,0}7}%
}}}}
\put(758,-1636){\makebox(0,0)[lb]{\smash{{\SetFigFont{8}{9.6}{\rmdefault}{\mddefault}{\updefault}{\color[rgb]{0,0,0}5}%
}}}}
\put(1944,-1351){\makebox(0,0)[lb]{\smash{{\SetFigFont{8}{9.6}{\rmdefault}{\mddefault}{\updefault}{\color[rgb]{0,0,0}2}%
}}}}
\put(2439,-2213){\makebox(0,0)[lb]{\smash{{\SetFigFont{8}{9.6}{\rmdefault}{\mddefault}{\updefault}{\color[rgb]{0,0,0}4}%
}}}}
\put(3196,-1426){\makebox(0,0)[lb]{\smash{{\SetFigFont{8}{9.6}{\rmdefault}{\mddefault}{\updefault}{\color[rgb]{0,0,0}3}%
}}}}
\put(3601,-2416){\makebox(0,0)[lb]{\smash{{\SetFigFont{8}{9.6}{\rmdefault}{\mddefault}{\updefault}{\color[rgb]{0,0,0}8}%
}}}}
\end{picture}%

\end{minipage}}%
\quad
\subfigure[
A decomposition tree of the subgraph $G_p$ from \protect\subref{fig:samplegraph}.
Building this decomposition tree is the first step for {\fbp}($G_p$).
The crossing edges at each node are indicated by
lines between the two children. Thus, $\crossing((G_{p})_0)=\{(1,5),(6,7)\}$
and $\crossing((G_{p})_{101})=\{(4,8)\}$.
Observe that each edge in $G_p$ is a crossing edge for exactly one node in the decomposition tree.
]{
\figlabel{fullybalpart-fig-b}
\begin{minipage}[b]{0.55\textwidth}
\centering
\begin{picture}(0,0)%
\epsfig{file=figures/FullBalPart.pstex}%
\end{picture}%
\setlength{\unitlength}{2368sp}%
\begingroup\makeatletter\ifx\SetFigFont\undefined%
\gdef\SetFigFont#1#2#3#4#5{%
  \reset@font\fontsize{#1}{#2pt}%
  \fontfamily{#3}\fontseries{#4}\fontshape{#5}%
  \selectfont}%
\fi\endgroup%
\begin{picture}(6740,4079)(968,-5928)
\put(5026,-4786){\makebox(0,0)[lb]{\smash{{\SetFigFont{8}{9.6}{\familydefault}{\mddefault}{\updefault}(2,4)(2,8)}}}}
\put(4351,-2761){\makebox(0,0)[lb]{\smash{{\SetFigFont{8}{9.6}{\familydefault}{\mddefault}{\updefault}{\color[rgb]{0,0,0}(1,2) (4,7)}%
}}}}
\put(1351,-4861){\makebox(0,0)[lb]{\smash{{\SetFigFont{8}{9.6}{\familydefault}{\mddefault}{\updefault}{\color[rgb]{0,0,0}(1,6)}%
}}}}
\put(5573,-5836){\makebox(0,0)[lb]{\smash{{\SetFigFont{8}{9.6}{\familydefault}{\mddefault}{\updefault}{\color[rgb]{0,0,0}4}%
}}}}
\end{picture}%

\end{minipage}}
\medskip
\subfigure[
The red-blue array for $G_p$. The blue elements have a dark shade. The red elements have a light shade.
There is one blue element for each vertex in $G_p$. There is one red element for each outgoing edge
in $G_p$. Since element $1$ is adjacent to two edges in $\outer(G_p)$, there are two 
red elements after it  in the red-blue array.  
The figure indicates a  subarray containing 
half of the blue elements and half of the red elements to within one. 
The red-blue array is used to make the fully-balanced partition of $G_p$. 
Specifically, $G_{px}$ will contain vertices
$2$, $5$, $6$, and $7$ and 
$G_{py}$ will contain vertices 
$1$, $3$, $4$, and $8$. 
Partition $G_{px}$ inherits three outer edges from $G_p$, and 
partition $G_{py}$ inherits five  outer edges from $G_p$. 
This particular subarray means that two paths in the decomposition tree will be cut. 
One path, separating element $1$ from $6$, goes from node $(G_p)_{00}$ to the root. 
The other path, separating element $2$ from $4$, goes from node $(G_p)_{10}$ to the root. 
The edges that are cut by this partition are the crossing edges of these nodes, i.e.,
$E(G_{px},G_{py}) = 
\{(1,6), (1,5), (6,7), (1,2), (4,7), (2,3), (3,8), (2,4), (2,8) \}$.
If $G_p$ is a node in the fully-balanced decomposition tree, then its left child 
will be 
$G_{px}$ and its right child will be $G_{py}$. 
]{
\figlabel{FullBalPart}
\begin{minipage}[b]{0.9\textwidth}
 \begin{center}
\begin{picture}(0,0)%
\epsfig{file=figures/decomptree_layout.pstex}%
\end{picture}%
\setlength{\unitlength}{2960sp}%
\begingroup\makeatletter\ifx\SetFigFont\undefined%
\gdef\SetFigFont#1#2#3#4#5{%
  \reset@font\fontsize{#1}{#2pt}%
  \fontfamily{#3}\fontseries{#4}\fontshape{#5}%
  \selectfont}%
\fi\endgroup%
\begin{picture}(9316,924)(668,-8848)
\put(8581,-8461){\makebox(0,0)[lb]{\smash{{\SetFigFont{11}{13.2}{\familydefault}{\mddefault}{\updefault}{\color[rgb]{0,0,0}3}%
}}}}
\put(2573,-8461){\makebox(0,0)[lb]{\smash{{\SetFigFont{11}{13.2}{\familydefault}{\mddefault}{\updefault}{\color[rgb]{0,0,0}6}%
}}}}
\put(773,-8461){\makebox(0,0)[lb]{\smash{{\SetFigFont{11}{13.2}{\familydefault}{\mddefault}{\updefault}{\color[rgb]{0,0,0}1}%
}}}}
\put(4366,-8461){\makebox(0,0)[lb]{\smash{{\SetFigFont{11}{13.2}{\familydefault}{\mddefault}{\updefault}{\color[rgb]{0,0,0}5}%
}}}}
\put(4981,-8469){\makebox(0,0)[lb]{\smash{{\SetFigFont{11}{13.2}{\familydefault}{\mddefault}{\updefault}{\color[rgb]{0,0,0}7}%
}}}}
\put(6174,-8461){\makebox(0,0)[lb]{\smash{{\SetFigFont{11}{13.2}{\familydefault}{\mddefault}{\updefault}{\color[rgb]{0,0,0}2}%
}}}}
\put(6766,-8461){\makebox(0,0)[lb]{\smash{{\SetFigFont{11}{13.2}{\familydefault}{\mddefault}{\updefault}{\color[rgb]{0,0,0}4}%
}}}}
\put(7374,-8461){\makebox(0,0)[lb]{\smash{{\SetFigFont{11}{13.2}{\familydefault}{\mddefault}{\updefault}{\color[rgb]{0,0,0}8}%
}}}}
\end{picture}%

 \end{center}
\end{minipage}}

\caption{The steps of the algorithm {\fbp}($G_p$) run on a sample graph.}
\figlabel{fullybalpart-fig}
\end{center}
\end{figure}

We first establish the running time of $\fbp(G_{p})$.

\begin{lemma}
  Given a graph $G_p$ that is a subgraph of a well shaped mesh $G$,
  $\fbp(G_{p})$ runs in $\bigO{\graphsize{G_p}\log{\graphsize{G_p}}}$
  on a RAM, both in expectation and with high probability (i.e.,
  probability at least $1-1/\poly(\graphsize{G_p})$).  In the DAM and
  cache-oblivious models, $\fbp(G_{p})$ uses
  $\bigO{1+(\graphsize{G_p}/B)\log{(\graphsize{G_p}/M)}}$ memory
  transfers in expectation and
  $\bigO{1+\graphsize{G_p}\log{\graphsize{G_p}/B}}$ memory transfers
  with high probability.
\lemlabel{FullyBalancedPartitionCO}
\end{lemma}

\begin{proof}
According to~\thmref{mesh-decomposition-tree}, Step 1 of
\fbp($G_p$) (computing $T_{G_p}$) takes time
$\bigO{\graphsize{G_p}\log{\graphsize{G_p}}}$ on a RAM, both in
expectation and with high probability.  In the DAM and
cache-oblivious models, this steps requires
$\bigO{1+(\graphsize{G_p}/B)\log{(\graphsize{G_p}/M)}}$ memory
transfers in expectation and
$\bigO{1+\graphsize{G_p}\log{\graphsize{G_p}/B}}$ memory transfers
with high probability. Steps~2-4 of \fbp($G_p$) each require linear
scans of an array of size $O(\graphsize{G_p})$, and therefore are
dominated by Step 1.
\end{proof}

We next establish the correctness of $\fbp(G_{p})$.
In the following, let constant $b$ represent the maximum degree of mesh $G$.

\begin{lemma}
Given a well shaped mesh $G$ and a subgraph $G_p\subseteq G$, 
\fbp~generates a fully-balanced partition of $G_p$.
\label{lem:FullyBalancedPartitionGood}
\end{lemma}

\begin{proof}
By the way that  we generate $V_{px}$ and $V_{py}$, we have
$$
\sizeof{\sizeof{V_{py}}- \sizeof{V_{px}}}\leq 1\,.
$$
This is because the number of blue elements in the subarray is exactly
$\sizeof{V_{px}}$, and the number of blue elements within and without
the subarray differ by at most one.

We next show that
\begin{equation}
  \sizeof{\sizeof{\outgoing(G_{py})}-\sizeof{\outgoing(G_{px})}}\leq
  2b+1\,.
\eqlabel{fully-balanced-condition}
\end{equation}
To determine $\sizeof{\outgoing(G_{px})}$ and
$\sizeof{\outgoing(G_{py})}$, modify the subarray as follows.
Remove from the subarray any red elements at the beginning of the
subarray before the first blue element \emph{in} the subarray.
Then add to the subarray any red elements before the first blue
element \emph{after} the subarray.  The number of red elements now
in the subarray is $\sizeof{\outgoing(G_{px})}$ and the number of
red elements not in the subarray is $\sizeof{\outgoing(G_{py})}$.
This modification can only increase or decrease
$\sizeof{\outgoing(G_{px})}$ and $\sizeof{\outgoing(G_{py})}$ each
by $b$, establishing~\eqref{fully-balanced-condition}.

Now, following~\cite{Leighton82,BhattLe84}, we show that
\begin{equation}
E(G_{px},G_{py}) \leq c{\sizeof{V_p}}^{\alpha}(1+{\beta}^{\alpha})/(1-{\beta}^{\alpha}).
\eqlabel{num-edges-crossing-both-cuts}
\end{equation}
By selecting a subarray of the red-blue array, we effectively make two
cuts on the leaves of the decomposition tree $T_{G_p}$.  (The only
time when there is apparently a single cut is if the subarray is the
first half of the array.  In this case, the second cut separates the
first leaf from the last.)  Consider one of these cuts.  The
array is split between two consecutive leaves of $\tree_{G_p}$. Denote
by $P$ the root of the smallest subtree of $\tree_{G_p}$ containing these
two leaves; see \figref{FullBalPart}. We consider the upward path
$P,P_1,P_2,\ldots,G_p$ in the decomposition tree $\tree_{G_p}$ from
$P$ up to the root $G_p$ of $\tree_{G_p}$.  Each node in the
decomposition tree on this path is a subgraph of $G$ that is being
split into two pieces.

We now count the number of edges that get removed as a result of these
splits:
\begin{eqnarray}
\sizeof{\crossing(P)\cup\crossing(P_1)\cup\crossing(P_2)\cup\ldots\cup\crossing(G_p)}
  &\leq& \sum_{i=0}^{\log_{\beta}{\sizeof{V}}} c\left ( \sizeof{V}/\beta^{i}\right )^{\alpha} \nonumber \\
  &\leq& c{\sizeof{V}}^{\alpha}/(1-\beta^{\alpha})\,.
  \eqlabel{num-edges-crossing-one-cut}
\end{eqnarray}
As reflected in \eqref{num-edges-crossing-one-cut}, each node along
the path has a different depth, which gives a geometric series.

The number of edges that cross from $G_{px}$ to $G_{py}$,
$E(G_{px},G_{py})$, is the number of edges that get removed when both
cuts get made.  However, doubling \eqref{num-edges-crossing-one-cut}
overestimates $E(G_{px},G_{py})$ by an amount $|\crossing(G_p)|$ since the root
$G_p$ can only be cut once.  Thus, doubling
\eqref{num-edges-crossing-one-cut} and subtracting $|\crossing(G_p)|$,
we establish \eqref{num-edges-crossing-both-cuts}.
\end{proof}

\subsection*{Fully-Balanced Decomposition Trees}

A \defn{fully-balanced decomposition tree} of a graph $G$ is a
decomposition tree of $G$ where the partition of every node (subgraph)
in the tree is fully-balanced.

We build a fully-balanced decomposition tree $\fbt_G$ of $G$ recursively.
First we apply the algorithm \fbp on the root $G$ to find
the left and right children, $G_0$ and $G_1$.
We next recursively build the fully balanced decomposition tree rooted at $G_0$
and the fully-balanced decomposition tree rooted at~$G_1$.

\begin{theorem}[Fully-Balanced Decomposition Tree for a Mesh]%
\thmlabel{fully-balanced-tree}
A fully-balanced decomposition tree of a mesh $G$ of constant
dimension can be computed in time
$\bigO{\graphsize{G}\log^2{\graphsize{G}}}$ on a RAM both in
expectation and with high probability.  The fully-balanced
decomposition tree can be computed in the DAM and cache-oblivious
models using $\bigO{1+(\graphsize{G}/B)\log^2{(\graphsize{G}/M)}}$
memory transfers in expectation and $\bigO{1+(\graphsize{G}/B)
  (\log^2{(\graphsize{G}/M)} +\log{\graphsize{G}})}$ memory transfers
with high probability.
\end{theorem}

\begin{proof}
We first establish that the construction algorithm takes expected time
$\bigO{\graphsize{G}\log^2\graphsize{G}}$ on a RAM. By
\lemref{FullyBalancedPartitionCO}, for any node $G_p$ in the
decomposition tree, we need $O(\sizeof{G_p}\log \sizeof{G_p})$
operations to build the left and right children, $G_{p0}$ and $G_{p1}$, 
both in expectation and with probability at least $1-1/\poly(\sizeof{G_p})$.
Since the left and
right children, $\sizeof{G_{p0}}$ and the $\sizeof{G_{p1}}$, of every
node $G_p$ differ in size by at most $1$, $\fbt_G$ has
$\Theta(\log\sizeof{G})$ levels. If $\sizeof{G_p}$ denotes the size of a node at level $i$, then 
level $i$ has construction time 
$O(\sizeof{G}\log \sizeof{G_p})$. 
Thus, the construction-time bound follows by linearity of expectation.

We next establish that the construction algorithm uses
$\bigO{1+\graphsize{G}\log^2{(\graphsize{G}/M)/B}}$ expected memory
transfers in the DAM and cache-oblivious models. 
Because we build the decomposition tree recursively, we give
a recursive analysis. 
The base case is when a node $G_p$ has size less than $M$ while its
parent node is greater than $M$.  Then the cost to build the entire
subtree $T_{G_p}$ is only $O(M/B)$, because this is the cost to read
all blocks of $G_p$ into memory.  Said differently, once a node is a
constant fraction smaller than $M$, the cost to build the
fully-balanced decomposition tree is $0$ because all necessary memory
blocks already reside in memory. There are therefore
$\Theta(\log\sizeof{G} - \log M)$ levels of the fully-balanced
decomposition tree having nonzero construction cost. Each level uses
at most $O((\sizeof{G}/B)\log(\sizeof{G}/M))$ memory transfers.  Thus,
the time bounds follows by linearity of expectation.

We next establish the high-probability bounds.  In the following
analysis, we examine, for each node $G_p$ in the fully-balanced 
decomposition tree, the decomposition tree $T_{G_p}$ that is used 
to build that node. We then group the nodes of all the decomposition trees by size
and count the number of nodes in each group. 

As an example, suppose that $\graphsize{G}$ is a power of two and all splits are even. 
There is one node of size $\graphsize{G}$ ---
the root node of the decomposition tree $T_G$.  There are four nodes
of size $\graphsize{G}/2$ --- two nodes in $T_G$, one node in
$T_{G_0}$, and one node in $T_{G_1}$. There are 12 node of size
$\graphsize{G}/4$ --- four nodes in $T_G$, two nodes in $T_{G_0}$, two
node in $T_{G_1}$, and one node in each of $T_{G_{00}}$, $T_{G_{01}}$,
$T_{G_{10}}$, and $T_{G_{11}}$.  

In general, let group $i$ contain all
decomposition tree nodes having size in the range
$(\graphsize{G}/2^{i}, \graphsize{G}/2^{i-1}]$.
Then group $i$ contains $\Theta(i2^i)$ nodes.

Analyzing each group separately, we show that the construction algorithm takes time
$O(\sizeof{G}\log^2{\sizeof{G}})$ on a RAM with high probability.
First, consider the $\Theta(\log\sizeof{G})$ largest nodes (those most
expensive to build), i.e., those in the smallest cardinality groups.  As
analyzed in~\thmref{mesh-decomposition-tree}, building these nodes
takes time $O(\sizeof{G}\log{\sizeof{G}})$ with high probability.

We analyze the rest of the node constructions group by group.  Since
each group $i$ contains $\Theta(i 2^{i-1})$ nodes, each successive
group contains more nodes than the total number of nodes in all
smaller groups.  As a result, there are $\Omega(\log{\sizeof{G}})$
nodes in each of the rest of the groups.  Thus, by~\lemref{lognlemma}, the
time to build the rest of the tree with high probability is the same
as the time in expectation, which is
$O(\sizeof{G}\log^2{\sizeof{G}})$.  Thus, we establish
high-probability bounds on the running time.

We now show that the construction algorithm takes
$\bigO{1+(\graphsize{G}/B) (\log^2({\graphsize{G}/M})
  +\log{\graphsize{G}}) }$ memory transfers with high probability.
First consider the $\Theta(\log\sizeof{G})$ largest nodes (those most
expensive to build).  As analyzed in~\thmref{mesh-decomposition-tree},
building these nodes uses $O(1+\sizeof{G}\log{\sizeof{G}}/B)$ memory
transfers with high probability.
Now examine all remaining nodes.   We
consider each level separately.  Each group contains
$\Omega(\log{\sizeof{G}})$ nodes.  Thus, by~\lemref{lognlemma}, the
high-probability cost of building the decomposition trees for all
remaining nodes matches the expected cost, which is
$O(1+(\sizeof{G}/B) \log^2{(\sizeof{G}/M)})$ memory transfers.
Thus, with high probability, the
construction algorithm takes $\bigO{1+(\sizeof{G}/B)(\log\sizeof{G} +
  \log^2{(\sizeof{G}/M)})}$ memory transfers with high probability, as promised.
\end{proof}

\subsection*{\boldmath $k$-Way Partitions}

We observe one additional benefit of \thmref{fully-balanced-tree}.
In addition to providing a simpler and faster algorithm for
constructing fully-balanced decomposition trees, we also provide a
new algorithm for $k$-way partitioning, as described
in~\cite{KiwiSpielmanTeng}.
For any positive integer $k>1$, a \defn{$k$-way partition} of a graph
$G = (V,E)$, is a $k$-tuple $(V_1,V_2,\ldots, V_k)$ (hence
$(G_1,G_2,\ldots, G_k)$) such that $\cup_{1\leq i \leq k} V_i = V$ and
$V_i \cap V_j=\emptyset$ for $i\neq j, 1\leq i, j \leq k$.  For any
$\beta \geq 1$, $(V_1,V_2,\ldots, V_k)$ is a $(\beta,k)$-way partition
if $\sizeof{G_i} \leq \beta\lceil \sizeof{G}/k \rceil$, for all $i\in
\{1,\ldots, k\}$.  It has been shown in \cite{KiwiSpielmanTeng} that
every well shaped mesh in $d$ dimensions has a $(1+\epsilon,k)$-way
partition, for any $\epsilon> 0$, such that $\max_i\{\outer{(G_i)}\} =
O((\sizeof{G}/k)^{1-1/d})$.

We now describe our $k$-way partition algorithm of a well shaped mesh
$G$. The objective is to evenly divide leaves of a fully-balanced
decomposition tree of $G$ into $k$ parts such that their number of
vertices are the same within one. First build a fully-balanced 
decomposition tree. Now assign the first $\sizeof{V}/k$ leaves to 
$V_1$, the next $\sizeof{V}/k$ leaves to $V_2$, and so on. 

In fact, we can modify this approach so that it runs faster by 
observing that we need not build the complete fully-balanced 
decomposition tree. First build the top 
$\Theta(\log k)$ levels of the tree, so that there are $\poly(k)$ leaves. 
At most $k$ of these leaves need to be refined further, since the remaining 
leaves will all belong to a single group $V_{i}$.

Our $k$-way partition algorithm using fully-balanced decomposition trees is
incomparable to the algorithm of~\cite{KiwiSpielmanTeng}. By building
fully-balanced decomposition tree, even a partial one, our algorithm
is slower than the algorithm of~\cite{KiwiSpielmanTeng}, which uses
geometric separators for partitioning instead. On the other hand, it
can be used to divide the nodes into $k$ sets whose sizes are
equal to within an additive one, instead of only asymptotically
the same size as in~\cite{KiwiSpielmanTeng}.


\secput{layout}{Cache-Oblivious Layouts}

In this section we show how to find a cache-oblivious layout of a mesh
$G$.  Given such a layout, we show that a mesh update runs
asymptotically optimally in $\Theta(1+\sizeof{G}/B)$ memory transfers
given the tall cache assumption that $M=\Omega(B^d)$. We also analyze
the performance of a mesh update when $M=o(B^d)$, bounding the
performance degradation for smaller $M$.

The layout algorithm is as follows.

\medskip

\begin{kbalgbox}%
\label{alg:layout}
{\sf CacheObliviousMeshLayout}($G$)
\begin{enumerate}
\item Build a $f(N)=O(N^{1-1/d})$ fully-balanced decomposition tree
  $T_G$ of $G$, as described in \thmref{fully-balanced-tree}.

\item Reorder the vertices in $G$ according to the order of the leaves
  in $T_{G}$. (Recall that each leaf in $T_{G}$ stores a single vertex
  in $G$.)  This reorder means: (a)~assign new indices to all vertices
  in the mesh, and (b)~for each vertex, let all neighbor vertices know
  the new index.

\end{enumerate}
\end{kbalgbox}%

We now describe the mechanics of relabeling and reordering.  Each
vertex knows its ordering and location in the input layout; this is
the vertex's index. A vertex also knows the index of each of its
neighboring vertices.  When we change a vertex's index, we apprise all
neighbor vertices of the change.  These operations entail a small
number of scans and cache-oblivious
sorts~\cite{FrigoLePrRa99,Prokop99,BrodalFaVi-08,BrodalFa02a}, for a
total cost of $O((\sizeof{G}/B)\log_{M/B}(\sizeof{G}/B)$ memory
transfers.  This cost is dominated by the cost to build the
fully-balanced decomposition tree.  (Thus, a standard merge sort,
which does not minimize the number of memory transfers, could also be
used.)

The cleanest way to explain is through an
example. Suppose that we have input graph
$
G=\set{ \set{a,b,c,d}, \set{(a,c), (a,d),(b,c),(c,d)} }
$,
which is laid out in input order:
$$
(a,c), (a,d), (b,c), (c,a), (c,b), (c,d), (d,a), (d,c)\,.
$$
Suppose that the leaves of fully-balanced decomposition tree are in
the order of $a,c,d,b$. This means that the renaming of nodes is as
follows: $[a:1],[c:2],[d:3],[b:4]$. (For clarity, we change indices
from letters to numbers.)  We obtain the reverse mapping $[a:1],
[b:4], [c:2], [d:3]$ by sorting cache-obliviously. We change the
labels on the first component of the edges by array scans:
$$
(a=1,c), (a=1,d), (b=4,c), (c=2,a), (c=2,b), (c=2,d), (d=3,a),
(d=3,c)\,.
$$
We then sort the edges by the second component,
$$
(c=2,a), (d=3,a), (c=2,b), (a=1,c), (b=4,c), (d=3,c), (a=1,d), (c=2,d)
\,,
$$
and change the labels on the second component of the edge by another
scan:
\begin{eqnarray*}
&&(c=2,a=1), (d=3,a=1),  (c=2,b=4), (a=1,c=2), (b=4,c=2), (d=3,c=2), \\ 
&&(a=1,d=3),  (c=2,d=3) \,.
\end{eqnarray*}
We get
$$
(2,1), (3,1),  (2,4), (1,2), (4,2), (3,2), (1,3),  (2,3) \,.
$$
We sort these edges by the first component to get the final layout.
The final layout is
$$
(1,2),  (1,3), (2,1), (2,3),  (2,4), (3,1),  (3,2), (4,2) \,.
$$

Thus, we obtain the following layout performance:
\begin{theorem}
  A cache-oblivious layout of a well shaped mesh $G$ can be computed
  in $O(\sizeof{G}\log^2\sizeof{G})$ time both in expectation and with
  high probability.  The cache-oblivious layout algorithm uses
  $O(1+\sizeof{G}\log^2(\sizeof{G}/M)/B)$ memory transfers in
  expectation and
  $O(1+(\sizeof{G}/B)(\log^2(\sizeof{G}/M)+\log\sizeof{G}))$ memory
  transfers with high probability.
\label{thm:cache-oblivious-layout-algo}
\end{theorem}

With such a layout, we can perform a mesh update cache-obliviously.
\begin{theorem}%
  Every well shaped mesh $G$ in $d$ dimensions has a layout that
  allows the mesh to be updated cache-obliviously with
  $O(1+\sizeof{G}/B)$ memory transfers on a system with block size $B$
  and cache size $M=\Omega(B^d)$.
\label{thm:mainCO}
\end{theorem}

\begin{proof}
We apply the algorithm described above on $G$ to build the layout.
Since each vertex of $G$ has constant degree bound $b$, its
size is bounded by a constant. Consider a row of nodes
$G_{p_1},G_{p_2},G_{p_3}\ldots$ in $T_{G}$ at a level such that each
node $G_{p_i}$ uses $\Theta(M)<M$ space and therefore fits in a
constant fraction of memory.

In the mesh update, the nodes of $G$ are updated in the order of
the layout, which means that first the vertices in $G_{p_1}$ are
updated, then vertices of $G_{p_2}$, then vertices of $G_{p_3}$,
etc. To update vertices of $G_{p_i}$, the vertices must be brought
into memory, which uses at most $O(1+M/B)$
 memory transfers.  In the mesh update, when we update a vertex $u$,
we access $u$'s neighbors. If the neighbor $v$ of $u$ is also in
$G_{p_i}$, i.e., edge $(u,v)$ is internal to $G_{p_i}$, then
accessing this neighbor uses no extra memory transfers. On the
other hand, if the neighbor $v$ is not in $G_{p_i}$,
then
following this edge requires another transfer hence an
extra block to be read into memory.

We now show that $\outer{(G_{p_i})}=
O(\sizeof{G_{p_i}}^{1-1/d})$. Since all subgraphs at the same level of
the fully-balanced decomposition tree are of the same size within one,
and outgoing edges of any subgraph are evenly split, each $G_{p_i}$
has roughly the same number of outer edges.  Suppose $G_{p_i}$ is in
level $j$. The total number of their outer edges are at most
$$
\sizeof{G}^{{1-1/d}} + 2\left(\frac{\sizeof{G}}{2}\right)^{{1-1/d}}  +
4 \left(\frac{\sizeof{G}}{4}\right)^{{1-1/d}}  + \ldots +
2^j \left(\frac{\sizeof{G}}{2^j}\right)^{{1-1/d}}  
\leq
\left(\frac{2^j}{2^{1/d} - 1}\right)
\left(\frac{\sizeof{G}}{2^j}\right)^{{1-1/d}} .
$$

Hence, $\outer{(G_{p_i})}=
O(\left(\sizeof{G}/{2^j}\right)^{{1-1/d}}) =
O(\sizeof{G_{p_i}}^{1-1/d})=O(M^{1-1/d})$. Therefore the total
size of memory that we need to perform a mesh update of the
vertices in $G_{p_i}$ is $\Theta(M+BM^{1-1/d})$.

By the tall-cache assumption that $B^d\leq M$, i.e., $B\leq M^{1/d}$,
and for a proper choice of constants, the mesh update for $G_{p_i}$
only uses $\Theta(M)<M$ memory. Since updating each node $G_{p_i}$ of
size $\Theta(M)$ uses $O(1+M/B)$ memory transfers, and there are a
total of $O(\sizeof{G}/M)$ such nodes, the update cost is
$O(1+\sizeof{G}/B)$, which matches the scan bound of $G$, and is
optimal.
\end{proof}

Thus, for dimension $d=2$, we have the ``standard'' tall-cache
assumption~\cite{FrigoLePrRa99}, and for higher dimensions we have a
more restrictive tall-cache assumption. We now analyze the tradeoff
between cache height and complexity.  Suppose instead of a cache with
$M=\Omega(B^d)$, the cache is of $M=\Omega(B^{d-\epsilon})$. We assume
$\epsilon < d-1$. We show that the cache performance of mesh update is
$B^{\epsilon/d}$ away from optimal.

\begin{corollary}
Every well shaped mesh $G$ in $d$ dimensions has a vertex ordering
that allows the mesh to be updated cache-obliviously with
$O(1+\sizeof{G}/B^{1-\epsilon/d})$ memory transfers on a system
with block size $B$ and cache size $M=\Omega(B^{d-\epsilon})$.
\end{corollary}
\begin{proof}
We apply similar analysis to that in \thmref{mainCO} on $G$.
From \thmref{mainCO}, the total size
of memory that we need to update mesh
$G_{p_i}$ is $\Theta(M+BM^{1-1/d})$. Since $M=\Omega(B^{d-\epsilon})$, we
have
\begin{eqnarray*}
O(M+BM^{1-1/d})&=& O(M+M\frac{B}{M^{1/d}}) \\
&\leq&
O(M+M\frac{B}{B^{1-\epsilon/d}})\\
 &=& O(M + MB^{\epsilon/d})\,.
\end{eqnarray*}
Thus, updating $G_{p_i}$ uses $O(1+ (M + MB^{\epsilon/d})/B)$ memory
transfers, which simplifies to $O(1 + \sizeof{G}/B^{1-\epsilon/d})$
memory transfers.
\end{proof}



\secput{improved}{\FLB Decomposition Trees and Faster Cache-Oblivious Layouts}

In this section we give the main result of this paper, a faster
algorithm for finding a cache-oblivious mesh layout of a well-shaped
mesh.  The main idea of the algorithm is to construct a new type of
decomposition tree, which we call a \defn{\rb decomposition tree}. The \rb
decomposition tree is based on what we call a \defn{\rb partition}. We give
an algorithm for building an \rb decomposition tree whose performance
is nearly a logarithmic factor faster than the algorithm for building
a fully-balanced decomposition tree.
We prove that an asymptotically optimal cache-oblivious mesh layout
can be found by traversing the leaves of the \rb decomposition tree.

\subsection*{\RB Partitions}

We first define the \rb partition of a subgraph $G_p$ of $G$. A
\defn{\rb $f$-partition} of $G_p\subseteq G$ is a
partitioning of $G_p$ into two subgraphs $G_{p0}$ and $G_{p1}$ such
that
\begin{itemize}
    \item $\sizeof{\crossing(G_{p})}\leq f(\sizeof{G_{p}})$,
    \item $\sizeof{G_{p0}} = \sizeof{G_{p1}} \pm
      \bigO{\sizeof{G_{p}}/\log^3{\sizeof{G}}}$,
         and
    \item  $\sizeof{\outgoing(G_{p0})}=\sizeof{\outgoing(G_{p1})}%
      \pm \bigO{\sizeof{\outgoing{(G_{p1})}}/\log^2{\sizeof{G}}}$.%
\end{itemize}

We next present an algorithm, \rbp, for computing \-balanced
partitions.  Given $G_p\subseteq G$, and a $(f(N)=O(N^\alpha),
\beta)$-partitioning geometric separator, {\rbp}($G_p$) computes a
\rb $(f(N)=O(N^\alpha))$-partition $G_{px}$ and $G_{py}$ of
$G_p$.

We find the \rb partition by building what we call a \defn{relax
  partition tree} $\tree_{G_p}$.
We call the top $3\log_{1/\beta}\log{\sizeof{G}}$ levels of
$\tree_{G_p}$ the~\defn{upper tree} of $\tree_{G_p}$ and the remaining
levels the \defn{lower tree} of $\tree_{G_p}$.

We build the upper tree by building the top
$3\log_{1/\beta}\log{\sizeof{G}}$ levels of a decomposition tree of
$G_p$.  By construction, all leaves of the upper tree (subgraphs of
$G_p$) contain at most $\sizeof{G_p}/\log^3{\sizeof{G}}$ vertices.
Outer edges of $G_p$ are distributed among these leaves.  By a
counting argument, there are at most $\log^2{\sizeof{G}}$ leaves that
can contain more than $\sizeof{\outer(G_p)}/\log^2{\sizeof{G}}$ outer
edges of $G_p$.

For each upper-tree leaf having more than
$\sizeof{\outer(G_p)}/\log^2{\sizeof{G}}$ outer edges, we refine the
leaf by building a decomposition tree on it.  We do not refine the
other leaves of the upper tree.  The union of these decomposition
trees comprise the lower tree.

Relax partition tree $\tree_{G_p}$ has leaves at different
depths. Some leaves are subgraphs having a single vertex while others
may have up to $\sizeof{G}/\log^3{\sizeof{G}}$ vertices.
The tree is stored in the same format as a standard decomposition tree.
Thus, leaves of the relax partition tree that are not refined
contain vertices stored in an arbitrary order.
The relax partition tree $\tree_{G_p}$ of $G_p$ is just
 a decomposition tree if there are fewer than $\log^3{\sizeof{G}}$ vertices.

\medskip

\begin{kbalgbox}
\label{alg:relax-balanced-partition}
{\rbp}$(G_p)$
\begin{enumerate}
\item \emph{Build $\tree_{G_p}$} --- Build the relax partition tree
  $\tree_{G_p}$ from $G_p$ recursively.

\item \emph{Build red-blue array} --- Build an array of vertices by an
  in-order traversal of leaves of $T_{G_p}$. Vertices in leaves that
  are not refined are laid out arbitrarily.  Build a red-blue array
  and find a subarray in the red-blue array as described in \fbp.

\item \emph{Modify red-blue array} -- Modify the subarray to satisfy
  the following constraint.  All vertices in an (unrefined) leaf must
  stay together, either within or without the subarray.  If any cut
  separates the vertices, then move the cut leftward or rightward to
  be in between the leaf node and a neighbor. Now partition the
  vertices in $G_{p}$ based on this modified subarray. Put the vertices
  representing blue elements that are in the subarray into set
  $V_{px}$ and put the vertices representing blue elements that are
  outside of the subarray into set $V_{py}$.

\item \emph{Partition $G_{p}$} --- Compute $G_{px}$ and $G_{py}$ from
  $V_{px}$ and $V_{py}$. This computation also means scanning edges to
  determine which edges are internal to $G_{px}$ and $G_{py}$ and
  which have now become external.
\end{enumerate}
\end{kbalgbox}


\begin{figure}[t!]
\begin{center}
\subfigure[An example subgraph $G_p$ of mesh $G$.
Subgraph $G_p$ has eight vertices,  ten edges, and eight outer edges
(i.e., $\sizeof{\outer(G_p) }$ = 8).
]{
\label{fig:samplegraphrbp} 
\begin{minipage}[b]{0.41\textwidth}
\centering
\begin{picture}(0,0)%
\epsfig{file=figures/samplegraph.pstex}%
\end{picture}%
\setlength{\unitlength}{2565sp}%
\begingroup\makeatletter\ifx\SetFigFont\undefined%
\gdef\SetFigFont#1#2#3#4#5{%
  \reset@font\fontsize{#1}{#2pt}%
  \fontfamily{#3}\fontseries{#4}\fontshape{#5}%
  \selectfont}%
\fi\endgroup%
\begin{picture}(3549,2307)(439,-3523)
\put(1576,-2311){\makebox(0,0)[lb]{\smash{{\SetFigFont{8}{9.6}{\rmdefault}{\mddefault}{\updefault}{\color[rgb]{0,0,0}1}%
}}}}
\put(1329,-3091){\makebox(0,0)[lb]{\smash{{\SetFigFont{8}{9.6}{\rmdefault}{\mddefault}{\updefault}{\color[rgb]{0,0,0}6}%
}}}}
\put(2049,-3316){\makebox(0,0)[lb]{\smash{{\SetFigFont{8}{9.6}{\rmdefault}{\mddefault}{\updefault}{\color[rgb]{0,0,0}7}%
}}}}
\put(758,-1636){\makebox(0,0)[lb]{\smash{{\SetFigFont{8}{9.6}{\rmdefault}{\mddefault}{\updefault}{\color[rgb]{0,0,0}5}%
}}}}
\put(1944,-1351){\makebox(0,0)[lb]{\smash{{\SetFigFont{8}{9.6}{\rmdefault}{\mddefault}{\updefault}{\color[rgb]{0,0,0}2}%
}}}}
\put(2439,-2213){\makebox(0,0)[lb]{\smash{{\SetFigFont{8}{9.6}{\rmdefault}{\mddefault}{\updefault}{\color[rgb]{0,0,0}4}%
}}}}
\put(3196,-1426){\makebox(0,0)[lb]{\smash{{\SetFigFont{8}{9.6}{\rmdefault}{\mddefault}{\updefault}{\color[rgb]{0,0,0}3}%
}}}}
\put(3601,-2416){\makebox(0,0)[lb]{\smash{{\SetFigFont{8}{9.6}{\rmdefault}{\mddefault}{\updefault}{\color[rgb]{0,0,0}8}%
}}}}
\end{picture}%

\end{minipage}}%
\quad
\subfigure[
A relax partition tree of the subgraph $G_p$ from 
\subref{fig:samplegraphrbp}.
Building this decomposition tree is the first step for {\rbp}($G_p$).
Observe that each edge in $G_p$ is a crossing edge for at most 
one node in the decomposition tree.
Some edges, such as $(2,4)$, are not crossing edges for any node. 
The top three levels of the decomposition tree are the upper tree. 
We refine a leaf of the upper tree if only it has many (at least three)
edges 
from $\outer{(G_p)}$. 
Upper tree leaf $(G_{p})_{00}$ has 4 edges from  $\outer{(G_p)}$.  
Upper tree leaf $(G_{p})_{01}$ has 1 edge  from  $\outer{(G_p)}$.   
Upper tree leaf $(G_{p})_{10}$ has 1 edge  from  $\outer{(G_p)}$. 
Upper tree leaf $(G_{p})_{11}$ has 2 edges  from  $\outer{(G_p)}$.
Thus, only $(G_{p})_{00}$ is further refined. 
]{
\begin{minipage}[b]{0.55\textwidth}
\centering
\begin{picture}(0,0)%
\epsfig{file=figures/RelaxBalPart.pstex}%
\end{picture}%
\setlength{\unitlength}{2368sp}%
\begingroup\makeatletter\ifx\SetFigFont\undefined%
\gdef\SetFigFont#1#2#3#4#5{%
  \reset@font\fontsize{#1}{#2pt}%
  \fontfamily{#3}\fontseries{#4}\fontshape{#5}%
  \selectfont}%
\fi\endgroup%
\begin{picture}(6740,3254)(968,-5103)
\put(4351,-2761){\makebox(0,0)[lb]{\smash{{\SetFigFont{8}{9.6}{\familydefault}{\mddefault}{\updefault}{\color[rgb]{0,0,0}(1,2) (4,7)}%
}}}}
\put(1351,-4861){\makebox(0,0)[lb]{\smash{{\SetFigFont{8}{9.6}{\familydefault}{\mddefault}{\updefault}{\color[rgb]{0,0,0}(1,6)}%
}}}}
\end{picture}%

\end{minipage}}
\medskip
\subfigure[
The red-blue array for $G_p$. 
The blue elements have a dark shade. 
The red elements have a light shade.
There is one blue element for each vertex in $G_p$. 
There is one red element for each outgoing edge in $G_p$. 
The figure indicates a subarray containing
half of the blue elements and half of the red elements to within one.
However, this subarray separates element $8$ from element $2$. 
This cut is not allowed because $8$ and $2$ are in the same 
leaf of the relax partition tree. Instead the cut is moved 
left to the first valid position. 
The new cut 
separates element $5$ from element $8$, which is allowed
because $5$ and $8$ are in different leaves 
of the relax partition tree. 
Thus, 
$G_{px}$ will contain vertices
$5$, $6$, and $7$, and
$G_{py}$ will contain vertices
$1$, $2$, $3$, $4$, and $8$.
]{
\label{fig:RelaxBalPart}
\begin{minipage}[b]{0.9\textwidth}
 \begin{center}
\begin{picture}(0,0)%
\epsfig{file=figures/RBPartLayout.pstex}%
\end{picture}%
\setlength{\unitlength}{2960sp}%
\begingroup\makeatletter\ifx\SetFigFont\undefined%
\gdef\SetFigFont#1#2#3#4#5{%
  \reset@font\fontsize{#1}{#2pt}%
  \fontfamily{#3}\fontseries{#4}\fontshape{#5}%
  \selectfont}%
\fi\endgroup%
\begin{picture}(9316,924)(668,-8848)
\put(2573,-8461){\makebox(0,0)[lb]{\smash{{\SetFigFont{11}{13.2}{\familydefault}{\mddefault}{\updefault}{\color[rgb]{0,0,0}6}%
}}}}
\put(773,-8461){\makebox(0,0)[lb]{\smash{{\SetFigFont{11}{13.2}{\familydefault}{\mddefault}{\updefault}{\color[rgb]{0,0,0}1}%
}}}}
\put(8581,-8461){\makebox(0,0)[lb]{\smash{{\SetFigFont{11}{13.2}{\familydefault}{\mddefault}{\updefault}{\color[rgb]{0,0,0}3}%
}}}}
\put(4381,-8461){\makebox(0,0)[lb]{\smash{{\SetFigFont{11}{13.2}{\familydefault}{\mddefault}{\updefault}{\color[rgb]{0,0,0}7}%
}}}}
\put(5573,-8461){\makebox(0,0)[lb]{\smash{{\SetFigFont{11}{13.2}{\familydefault}{\mddefault}{\updefault}{\color[rgb]{0,0,0}5}%
}}}}
\end{picture}%

 \end{center}
\end{minipage}}
\caption{The steps of the algorithm {\rbp}($G_p$) run on a sample graph.}
\label{fig:RBPartLayout}
\end{center}
\end{figure}

We first establish the running time of {\rbp}$(G_p)$.
\begin{lemma}
  Given a subgraph $G_p$ of a well shaped mesh $G$, $\sizeof{G_p} \geq
  \log^3\sizeof{G}$, $\rbp(G_{p})$ runs in time
  $\bigO{\sizeof{G_p}\log\log{\sizeof{G}}}$ on a RAM and
  $\bigO{1+(\sizeof{G_p}/B)\min\{\log\log{\sizeof{G}},
    \log(\sizeof{G_p}/M)\}}$ memory transfers in the DAM and
  cache-oblivious models in expectation.  With high probability, it
  runs in $\bigO{\sizeof{G_p}\log{\sizeof{G_p}}}$ on a RAM and
  $\bigO{1+\sizeof{G_p}\log{\sizeof{G_p}}/B}$ memory transfers in the
  DAM and cache-oblivious models.
\lemlabel{RelaxBalancedPartitionCO}
\end{lemma}

\begin{proof}
We establish that the construction algorithm runs in expected time
$\bigO{\graphsize{G_p}\log\log\graphsize{G}}$ on a RAM.  The upper
tree of $\tree_{G_p}$ takes expected time
$O(\sizeof{G_p}\log\log{\sizeof{G}})$.  There are at most
$\log^2{\sizeof{G}}$ leaves of the upper tree to be refined.  For
each of these leaves, we build a decomposition tree, and this takes
expected time
$$
O((\sizeof{G_p}/\log^3{\sizeof{G}})\log(\sizeof{G_p}/\log^3{\sizeof{G}}))
\leq
O(\sizeof{G_p}/\log^2{\sizeof{G}}).
\label{eq:log-cube-sum}
$$
Thus, the total expected time to refine all leaves is $O(\sizeof{G})$.
Steps 2-4 takes linear time.  Thus, {\rbp}$(G_p)$ finds a \flb
partition in expected time $\bigO{\sizeof{G_p}\log\log{\sizeof{G}}}$.

We next establish that the construction algorithm uses
$\bigO{1+(\graphsize{G_p}/B)\min\{\log\log{\sizeof{G}},
  \log(\sizeof{G_p}/M)\}}$ expected memory transfers in the DAM and
cache-oblivious models. There are two cases. The first case is when
$M\geq\sizeof{G_p}/\log^3{\sizeof{G}}$.  Then some of nodes in the top
$3\log_{1/\beta}\log{\sizeof{G}}$ levels of the $\tree_{G_p}$ may be a
constant fraction smaller than $M$. Such small nodes require no memory
transfers to build, because they are already stored in memory.  Only
the top $O(\log(\sizeof{G_p}/M))$ levels use memory transfers. The
rest of the decompositions are free of memory transfers because all
necessary memory blocks already reside in memory. When a subgraph
$G_{p}$ has size $\Omega(M)$, then the partition of a subgraph takes
expected $\Theta(\sizeof{G_p}/B)$ memory transfers, because this is
the cost of a linear scan. Hence, the total cost is
$O(1+(\sizeof{G_p}/B)\log(\sizeof{G_p}/M))$.

The second case is when $M<\sizeof{G_p}/\log^3{\sizeof{G}}$. Then, the
upper tree of $\tree_{G_p}$ takes
$O(1+\sizeof{G_p}\log\log{\sizeof{G}}/B)$ memory transfers in
expectation.  There are at most $\log^2{\sizeof{G}}$ leaves of the
upper tree of $\tree_{G_p}$ that need further refinement, and the leaf
sizes are at most $\sizeof{G_p}/\log^3{\sizeof{G}}$.  Building a
decomposition tree on one of these leaves takes
$$
\bigO{1+(\sizeof{G_p}/\log^3{\sizeof{G}})
\log(\sizeof{G_p}/\log^3{\sizeof{G}})/B}
\leq \bigO{1+\sizeof{G_p}/B \log^2{\sizeof{G}}}
$$
memory transfers in expectation.
Since there are at most  $\log^2{\sizeof{G}}$
leaves, the total expected number of memory transfers to construct
the lower tree of  $\tree_{G_p}$ is
$\bigO{\sizeof{G_p}/B}$, which is dominated by the cost to build the
upper tree.

Combining the two cases, we obtain that the expected number of memory
transfers to build $\tree_{G_p}$ is
$O(1+(\sizeof{G_p}/B)\min\{\log\log{G}, \log(\sizeof{G_p}/M)\})$.

We next establish the high-probability bounds. We first consider all
nodes that have size $\Omega(\sizeof{G_p}/\log\sizeof{G_p})$.  There
are $\bigO{\log{\sizeof{G_p}}}$ such nodes. Building these nodes uses
time $O(\sizeof{G_p}\log{\sizeof{G_p}})$ and
$O(1+\sizeof{G_p}\log{\sizeof{G_p}}/B)$ memory transfers with high
probability by~\thmref{mesh-decomposition-tree}.

For the rest of the \emph{upper tree} of $\tree_{G_p}$, each level
contains $\Omega(\log \sizeof{G_p})$ nodes. Thus, the number of
memory transfers with high probability matches the number of
memory transfers in expectation, which is
$O(1+(\sizeof{G_p}/B)\min\{\log\log{G}, \log(\sizeof{G_p}/M)\})$.
The cost to build the rest of the upper tree
is dominated by the cost to build the largest
$\bigO{\log{\sizeof{G_p}}}$ nodes in the upper tree.

As described above, the expected cost to build the lower tree is
$\bigO{\sizeof{G_p}}$ time and $\bigO{\sizeof{G_p}/B}$ memory
transfers. The high-probability bounds are at most a
$O(\log\sizeof{G_p})$ factor greater and hence are dominated by the
cost to build the upper tree. Thus, we establish the high probability
bounds on time and memory transfers.
\end{proof}

We next establish the correctness of $\rbp(G_{p})$. In the following,
let $b$ represent the maximum degree of mesh $G$.

\begin{lemma}
Given a well shaped mesh $G$ and a subgraph $G_p\subseteq G$, {\rbp}$(G_p)$
generates a \rb  partition of $G_p$.
\lemlabel{RelaxBalancedPartitionGood}
\end{lemma}

\begin{proof}
By the way we construct the relax partition tree $\tree_{G_p}$,
nodes that are not refined contain
$O(\sizeof{\outer(G_p)}/\log^2{\sizeof{G}})$ outer edges of $G_p$,
and their sizes differ by $O(\sizeof{G_p}/\log^3{\sizeof{G}})$.
 Thus, by the way we generate $G_{px}$ and
$G_{py}$, the number of outgoing edges of $G_{px}$ and $G_{py}$
differ by $O(\sizeof{\outer{(G_p)}}/\log^2{\sizeof{G}})$ and
$\sizeof{G_{px}}$ and $\sizeof{G_{py}}$ differ by
$O(\sizeof{G_p}/\log^3{\sizeof{G}})$. Recall that
$\outgoing{(G_{px})} \cup \outgoing{(G_{py})} = \outer{(G_{p})}$.
Thus, we have
$$\sizeof{\outgoing{(G_{px})}} =
\sizeof{\outgoing{(G_{py})}} \pm
O(\sizeof{\outgoing{(G_{py})}}/\log^2{\sizeof{G}}).
$$

As shown in Equation~\eqref{num-edges-crossing-both-cuts} from
\lemref{FullyBalancedPartitionGood}, the number of crossing edges
satisfies $ \sizeof{\crossing(G_{p})}\leq f(\sizeof{G_{p}}). $
\end{proof}

\subsection*{\RB Decomposition Trees}

A \defn{\flb decomposition tree} of a well shaped mesh $G$
is a decomposition tree of $G$ where every partition
of every node $G_p$  in the tree is \flb.

We construct a \flb decomposition tree of $G$ recursively.
First we apply the algorithm \rbp on the root $G$ to get
the left and right children, $G_0$ and $G_1$.
We next recursively build the (left) subtree rooted at $G_0$
and then the (right) subtree rooted at $G_1$.

\begin{theorem}[\protect\FLB Decomposition Tree for a Mesh]%
  \thmlabel{flb-tree}
A \flb decomposition tree of a well shaped mesh $G$ of constant dimension
can be computed in time
$\bigO{\graphsize{G}\log\graphsize{G}\log{\log{\graphsize{G}}}}$
on a RAM both in expectation and with high probability.
The \flb decomposition tree can be computed in the DAM and cache-oblivious
models using
$\bigO{1+(\graphsize{G}/B)\log{(\graphsize{G}/M)}\min\{\log\log{\sizeof{G}},\log(\sizeof{G}/M)\} }$
memory transfers in expectation and
$\bigO{1+(\graphsize{G}/B)(\log{(\graphsize{G}/M)}\min\{\log\log{\sizeof{G}},\log(\sizeof{G}/M) \} + \log{\graphsize{G}})}$
memory transfers with high probability.
\end{theorem}

\begin{proof}
  When $\sizeof{G} \leq M$, the construction algorithm takes
  $O(\sizeof{G})$ time and $O(\sizeof{G}/B)$ memory transfers, both in
  expectation and with high probability. We consider 
  $O(\sizeof{G})=\Omega(M)$ in the following analysis.

  We first analyze the expected running time of the algorithm on a
  RAM.  The construction time of each node $G_p$ is
  $\bigO{\sizeof{G_p}\log\log\sizeof{G}}$, and there are
  $\bigO{\log{\sizeof{G}}}$ levels in the \rb decomposition tree.
  Thus, by linearity of expectation, the expected running time is
  $\bigO{\sizeof{G}\log\sizeof{G}\log\log{\sizeof{G}}}$.

  We show that the construction algorithm uses
  $\bigO{1+(\sizeof{G}/B)\log(\sizeof{G}/M)\min\{\log\log{\sizeof{G}},
    \log(\sizeof{G}/M)\}}$ memory transfers in the DAM and
  cache-oblivious models.  We analyze large and small nodes in the \rb
  decomposition tree differently.  There are two cases.  The first
  case is when a tree node $G_{p}$ is large, i.e., $\sizeof{G_p}\geq
  \log^3{\sizeof{G}}$.  In this case, {\rbp}$(G_p)$ uses expected
  $\bigO{1+(\sizeof{G_p}/B)\min\{\log\log{\sizeof{G}},
    \log(\sizeof{G_p}/M)\}}$ memory transfers
  by~\lemref{RelaxBalancedPartitionCO}.  Since all nodes a constant
  factor smaller than $M$ can be constructed with no memory transfers,
  we only need consider nodes of size $\Omega{(M)}$.  There are
  $\bigO{\log(\sizeof{G}/M)}$ levels of nodes of size $\Omega{(M)}$.
So the construction of all nodes larger than  $\log^3\sizeof{G}$  takes
$\bigO{1+(\sizeof{G}/B)\log(\sizeof{G}/M)\min\{\log\log{\sizeof{G}},
  \log(\sizeof{G}/M)\}}$ expected memory transfers.

The second case is when $\sizeof{G_p}< \log^3{\sizeof{G}}$.  In this
case, we build a complete decomposition tree at each node.  Therefore
by~\lemref{FullyBalancedPartitionCO}, the cost to build one of these
nodes is $\bigO{1+(\graphsize{G_p}/B)\log{(\graphsize{G_p}/M)}}$ in
expectation.  As before, nodes a constant factor smaller than $M$ can
be constructed with no memory transfers.  Therefore, the number of
levels containing nodes of size between $\Omega(M)$ and less than
$\log^3{\sizeof{G}}$ is at most $\bigO{\log{(\log^3{\sizeof{G}}/M)}}$.
Thus, the construction of all nodes of size
$\bigO{\log^3{\sizeof{G}}}$ uses
$O(1+(\sizeof{G}/B)\log^2(\log^3{\sizeof{G}}/M))$ memory transfers in
expectation, which is dominated by the first case.

Now we establish the high probability bounds.  We analyze the largest
$\Theta(\log{\sizeof{G}})$ nodes and the remaining nodes of the \rb
decomposition tree separately.  Any level of the \rb decomposition
tree below the largest $\Theta(\log{\sizeof{G}})$ nodes has
$\Omega(\log{\sizeof{G}})$ nodes.  Hence, for each level, the
construction cost with high probability matches the construction cost
in expectation, which is $\bigO{\sizeof{G}\log\log{\sizeof{G}}}$
expected time and $\bigO{1+(\sizeof{G}/B)\min\{\log\log{\sizeof{G}},
  \log(\sizeof{G}/M)\}}$ expected memory transfers.  Since the
construction algorithm is recursive, a \rb partition of nodes a
constant fraction smaller than $M$ uses no memory transfers.  Hence,
all levels of the \rb decomposition tree other than the largest
$\Theta(\log{\sizeof{G}})$ nodes can be constructed in
$\bigO{\sizeof{G}\log\sizeof{G}\log\log{\sizeof{G}}}$ time in a RAM
and $\bigO{1+(\sizeof{G}/B)
  \log(\sizeof{G}/M)\min\{\log\log{\sizeof{G}}, \log(\sizeof{G}/M)\}}$
memory transfers with high probability.

For the largest $\Theta(\log{\sizeof{G}})$ nodes of the \rb decomposition 
tree, we establish the
high probability bounds using a different approach.  Similar to the
proof of~\thmref{fully-balanced-tree}, we examine each relax partition
tree that is used to build each node of the relax-balanced
decomposition tree, and we examine all nodes within all of these relax
partition trees.  However, now there are upper trees and lower trees;
we examine the nodes within upper and lower trees separately.

We look at the upper trees of the relax partition trees of the largest
$\Theta(\log\sizeof{G})$ nodes of the \rb decomposition tree.
There are $\Theta(\log\sizeof{G})$ upper trees, which
are complete binary trees.
Following a similar analysis
to that in the proof of~\thmref{fully-balanced-tree}, the
construction of  the largest $\Theta(\log\sizeof{G})$  nodes from among the
$\Theta(\log\sizeof{G})$ upper trees takes
$O(\sizeof{G}\log{\sizeof{G}})$ time and uses
$O(1+\sizeof{G}\log{\sizeof{G}}/B)$ memory transfers with high
probability.

For the rest of the nodes in the upper trees, the high probability
bounds match the expectation bounds, both in time and memory transfers
by~\thmref{fully-balanced-tree}.  Therefore building the nodes in the
rest of the upper trees takes $O(\sizeof{G}\log^2\log{\sizeof{G}})$
time and uses $O(\sizeof{G}\log^2\log{\sizeof{G}}/B)$ memory transfers
with high probability.  This cost is dominated by the construction
cost of the largest $\Theta(\log{\sizeof{G}})$ nodes of the upper
trees.

We now focus on the lower trees of the relax partition trees of the
largest $\Theta(\log\sizeof{G})$ nodes of the \rb decomposition tree.
We show that the cost to build all of the lower trees takes time
$O(\sizeof{G}\log{\sizeof{G}})$ and uses
$O(1+\sizeof{G}\log{\sizeof{G}}/B)$ memory transfers with high
probability (i.e., probability $1- 1/\poly(\sizeof{G})$).  With high
probability, the lower tree of a partition tree $\tree_{G_p}$ of a
subgraph $G_p$ can be computed in $O(\sizeof{G_p})$ on a RAM and with
$O(\sizeof{G_p}/B)$ memory transfers in the DAM and the
cache-oblivious models.  Given a node $G_p$ and its relax partition
tree $\tree_{G_p}$, there are two cases.  The first case is when there
are $\Omega(\log{\sizeof{G}})$ leaves of the upper tree of
$\tree_{G_p}$ that need to be refined.  Thus, with high probability,
the construction cost of the lower tree of $\tree_{G_p}$ matches the
expected construction cost, which is in $O(\sizeof{G_p})$ time and
$O(\sizeof{G_p}/B)$ memory transfers, as analyzed
in~\lemref{RelaxBalancedPartitionCO}.

The second case is when there are $O(\log{\sizeof{G}})$ leaves of the
upper tree of $\tree_{G_p}$ that need to be refined.  The construction
cost of a single leaf is $O(\sizeof{G_p}/\log^2{\sizeof{G}})$ time and
$O(\sizeof{G_p}/B\log^2{\sizeof{G}})$ memory transfers in expectation.
Thus, the construction cost to refine a single leaf with high
probability is $O(\sizeof{G_p}/\log{\sizeof{G}})$ time and
$O(\sizeof{G_p}/B\log{\sizeof{G}})$ memory transfers and the
construction cost to refine all leaves with high probability is
$O(\sizeof{G_p})$ time and $O(\sizeof{G_p}/B)$ memory transfers. Thus, all
lower trees of the relax partition trees of the largest
$\Theta(\log\sizeof{G})$ nodes of the \rb decomposition tree can be
constructed in $O(\sizeof{G})$ time and $O(\sizeof{G}/B)$ memory
transfers with high probability, which is dominated by the
construction of all upper trees.

Hence, with high probability, the construction algorithm runs in
$\bigO{\sizeof{G}\log\sizeof{G}\log\log{\sizeof{G}}}$ time on a RAM
and uses $\bigO{1 +
  (\sizeof{G}/B)(\log(\sizeof{G}/M)\min\{\log\log{\sizeof{G}},\log(\sizeof{G}/M)\}+
  \log\sizeof{G})}$ memory transfers in the DAM and the
cache-oblivious models.
\end{proof}

We now show that a \flb decomposition tree can serve the same purpose
as a fully-balanced decomposition tree in giving cache-oblivious
layout. The crucial property is the following.
\begin{lemma}%
\lemlabel{new-tree-prop} Given a \flb decomposition tree of
graph $G$, all nodes on any level of the \flb decomposition
tree contain the same number of vertices to within an $o(1)$
factor and all outgoing degrees are the same size to within an
$o(1)$ factor.
\end{lemma}

\begin{proof}
\sloppy
From the definition of \flb, for any subgraph $G_p$ and its two
children  $G_{p_0}$ and $G_{p_1}$
  $\sizeof{\outgoing(G_{p0})}=\sizeof{\outgoing(G_{p1})}
      \pm \bigO{\sizeof{\outgoing(G_{p1})}/\log^2{\sizeof{G}}}$,
      and
$\sizeof{G_{p0}} = \sizeof{G_{p1}} \pm
      \bigO{\sizeof{G_{p}}/\log^3{\sizeof{G}}}$.
Thus, for constant $c$, the ratio of the outgoing degree or the size
between any two subgraphs at depth $i$ is at most
$(1+c/\log^2{\sizeof{G}})^i$ and $(1+c/\log^3{\sizeof{G}})^i$. Since there are
$O(\log\sizeof{G})$ levels, these
quantities differ by at most an $o(1)$ factor, as promised.
\end{proof}

Similar to \secref{layout}, to find a cache-oblivious layout of a well
shaped mesh $G$, we build a \flb decomposition tree of $G$. The
in-order traversal of the leaves gives the cache-oblivious layout.
\lemref{new-tree-prop} guarantees that we can apply the same analysis
from \secref{layout} to show that we have a cache-oblivious layout.

We thus obtain the following result:

\begin{theorem}\label{thm:improved_mainCO}
A cache-oblivious layout
of a well shaped mesh $G$ can be computed in time
$\bigO{\graphsize{G}\log\graphsize{G}\log{\log{\graphsize{G}}}}$
on a RAM both in expectation and with high probability.
The cache-oblivious layout can be computed in the DAM and cache-oblivious
models using
$\bigO{1+(\graphsize{G}/B)\log{(\graphsize{G}/M)}
\min\{\log\log{\sizeof{G}}, \log(\sizeof{G}/M)\}}$
memory transfers in expectation and
$\bigO{1+(\graphsize{G}/B)(\log{(\graphsize{G}/M)} 
\min\{\log\log{\sizeof{G}}, \log(\sizeof{G}/M)\}+\log{G})}$
memory transfers with high probability.
\end{theorem}



\secput{related}{Applications and Related Work}

\subsection*{Applications of Mesh Update}
The mesh update problem appears in many scientific computations and
ranks among most basic primitives for numerical computations.  In
finite-element and finite-difference methods, one must solve very
large-scale sparse linear systems whose underlying matrix structures
are meshes \cite{MillerTengThurstonVavasis}.  In practice, these
linear systems are solved by conjugate gradient or preconditioned
conjugate gradient methods \cite{DemmelDoEiFuPeVuWhYe05,Simon}.  The
most computational intensive operation of conjugate gradient is a
matrix-vector multiplication operation
\cite{DemmelDoEiFuPeVuWhYe05,VuducDeYe06,Vuduc03,BenderBroFagJacVic07}
which amounts to a mesh update in finite-element applications. The
iterative conjugate gradient method repeatedly performs mesh updates.
Mesh update is also the key operation in fast multipole methods (FMM)
for N-body simulation \cite{GreengardRokhlin,HackneyEastwoodBook},
especially when particles are not uniformly distributed \cite{Teng98}.
The partitioning and layout techniques presented here also apply to
the adaptive quadtrees or octtrees used in non-uniform N-body
simulation.

\subsection*{Related Work}
The cache-oblivious memory model was introduced
in~\cite{FrigoLePrRa99,Prokop99}, and cache-oblivious algorithms have
been developed for many scientific problems such as matrix
multiplication, FFT, and LU
decomposition~\cite{FrigoLePrRa99,Prokop99,BlumofeJoKu96,Toledo97},
Now the area of cache-oblivious data structures and algorithms is a
lively field.

There are other approaches to achieving good locality in scientific
computations.  One alternative to the cache-oblivious approach is to
write self-tuning programs, which measure the memory system and adjust
their behavior accordingly. Examples include scientific applications
such as FFTW~\cite{FrigoJo98}, ATLAS~\cite{atlas_sc98}, and
self-tuning databases (e.g.,~\cite{WMHZ}).  The self-tuning approach
can be complementary to the cache-oblivious approach.  For example,
some versions of FFTW~\cite{FrigoJo98} begin optimization starting
from a cache-oblivious algorithm.

Methods exploiting locality for both sequential (out-of-core) and
 parallel implementation of iterative methods for sparse linear
   systems have long history in scientific computing.  Various
   partitioning algorithms have been developed for load balancing and
   locality on parallel machines
   \cite{Metis,Chaco,MillerTengThurstonVavasis,Simon}, and algorithms
   that have good temporal locality have been proposed and implemented
   for the out-of-core sparse linear solvers \cite{toledo95}.  A mesh
   update can be viewed as a sparse matrix-dense vector
   multiplication, and there exist upper and lower bounds on the I/O
   complexity of this primitive~\cite{BenderBroFagJacVic07}. However,
   these bounds apply to any type of matrix, whereas special structure
   of well-shaped meshes enables more efficient mesh updates.

Since the mesh-update problem is reminiscent of graph traversal, we
briefly summarize a few results in external-memory graph traversal.
The earliest papers in this area apply to general directed
graphs~\cite{BuchsbaumGoldwasserEtc,ChiangGoodrichEtc,AbelloBuchsbaumWestbrook,NodineGooVit96}
and others focus on more specialized graphs, such as planar directed
graphs~\cite{ArgeBrodalToma} or undirected graphs perhaps of bounded
degree~\cite{MunagalaRanade,Meyer2001,MehlhornMeyer,ArgeMeyTom04,ChowdhuryRam05}.
The problem of cache-oblivious graph traversal and related problems is
addressed
by~\cite{ArgeBenderDemaineHolland-Minkley,BrodalFagerbergMeyerZeh}.
There are also external-memory and cache-oblivious algorithms for
other common graph problems, but such citations are beyond the scope
of this paper.

The problem of cache-oblivious mesh layouts is first described
in~\cite{YoonLindstromPascucci}.  This paper gives no theoretical
guarantees either on the traversal cost or the cost to generate the
mesh layout, however. It does propose heuristics for mesh layout that
give good running times, in practice, for a range of types of mesh
traversals.


\end{document}